\documentclass[english,10pt]{article}
\usepackage[T1]{fontenc}
\usepackage[latin9]{inputenc}
\usepackage{verbatim}
\usepackage{amsthm}
\usepackage{amsmath}
\usepackage{amssymb}
\usepackage{esint}

\makeatletter
\numberwithin{equation}{section}
\numberwithin{figure}{section}
\theoremstyle{plain}
\newtheorem{thm}{Theorem}[section]
  \theoremstyle{remark}
  \newtheorem{rem}[thm]{Remark}
  \theoremstyle{plain}
  \newtheorem{prop}[thm]{Proposition}
 \theoremstyle{definition}
  \newtheorem{example}[thm]{Example}
  \theoremstyle{definition}
  \newtheorem{defn}[thm]{Definition}
  \theoremstyle{plain}
  \newtheorem{lem}[thm]{Lemma}
  \theoremstyle{plain}
  
  \theoremstyle{remark}
  \newtheorem*{acknowledgement*}{Acknowledgement}

\usepackage{amsfonts}
\usepackage{dsfont}
\usepackage{mathrsfs}
\usepackage{enumerate}
\usepackage[english,frenchb]{babel}
\usepackage{amssymb,amsmath}

\DeclareMathOperator*{\Tr}{Tr}

\makeatother

\usepackage{babel}

\begin{document}
\selectlanguage{english}

\title{Higher order terms for the quantum evolution of a Wick observable
within the Hepp method}

\author{Sébastien Breteaux%
\thanks{IRMAR, UMR-CNRS 6625, Université de Rennes 1, campus de Beaulieu,
35042 Rennes Cedex, France. E-mail: sebastien.breteaux@univ-rennes1.fr%
}}

\date{Februar 2011}
\maketitle
\begin{abstract}
\noindent The Hepp method is the coherent state approach to the mean
field dynamics for bosons or to the semiclassical propagation. A key
point is the asymptotic evolution of Wick observables under the evolution
given by a time-dependent quadratic Hamiltonian. This article provides
a complete expansion with respect to the small parameter $\varepsilon>0$
which makes sense within the infinite-dimensional setting and fits
with finite-dimensional formulae.
\end{abstract}


\selectlanguage{english}

\section{Introduction}

In this article we derive two expansions with respect to a small parameter
$\varepsilon$ of quantum evolved Wick observables under a time-dependent
quadratic Hamiltonian.

The Hepp method was introduced in~\cite{MR0332046} and then extended
in~\cite{MR530915,MR539736} in order to study the mean field dynamics
of many bosons systems via a (squeezed) coherent states approach.
The asymptotic analysis in the mean field limit is done with respect
to a small parameter~$\varepsilon$, where the number of particles
is of order~$\frac{1}{\varepsilon}$.

Remember that the mean field dynamics is obtained as a classical Hamiltonian
dynamics which governs the evolution of the center~$z(t)$ of the
Gaussian state (squeezed coherent state). Meanwhile the covariance
of this Gaussian as well as the control of the remainder term is determined
by the evolution of a quadratic approximate Hamiltonian around~$z(t)$.

A key point in this method is the asymptotic analysis of the evolution
of a Wick quantized observable according to this quantum time-dependent
quadratic Hamiltonian.

Only a few results are clearly written about the remainder terms and
some possible expansions in powers of $\varepsilon$, see the works
of Ginibre and Velo~\cite{MR602197,MR605198}. In the finite-dimensional
case, entering into the semiclassical theory, accurate results have
been given by Combescure, Ralston and Robert in~\cite{MR1690026}.
For the mean field infinite-dimensional setting some results have
been proved in~\cite{MR2575484,MR2291792,MR2530155} with
a different approach.

We stick here with the Hepp method with the presentation of~\cite{MR2465733}
which puts the stress on the similarities and differences between
the infinite-dimen\-sional bosonic mean field problem and the finite-dimensional
semiclassical analysis. Nevertheless, in~\cite{MR2465733} the authors
only considered the main order term although some of their formulae
make possible complete expansions. In this article we derive two expansions
of the quantum evolved Wick observables which are equal term by term.

Two difficulties have to be solved :
\begin{enumerate}
\item Unlike the time-independent finite-dimensional case, no Mehler type
explicit formula (see for example~\cite{MR1339714} or~\cite{MR883643})
is available. A general time-dependent Hamiltonian has no explicit
dynamics.
\item In the infinite-dimensional framework the quantization of a linear
symplectic transformation (a Bogoliubov transformation) requires some
care. Useful references on this subject are~\cite{MR0208930} and~\cite{MR1178936}.
Its realization in the Fock space relies on a Hilbert-Schmidt condition
on the antilinear part connected with the Shale theorem (see~\cite{MR0137504}
and~\cite{MR628382,MR2221699,MR2297950}).
\end{enumerate}
These things are well known but have to be considered accurately while
writing complete expansions.

Two different methods, with apparently two different final formulae,
will be used. A first one relies on a Dyson expansion approach and
provides the successive terms as time-dependent integrals. The second
one uses the exact formulae for the finite-dimensional Weyl quantization
and after having made explicit the relationship between Wick and Weyl
quantizations like in~\cite{MR1186643} or~\cite{MR2465733}, the
proper limit process with respect to the dimension is carried out.

The outline of this article is the following. In Section~\ref{sec:Wick-calculus}
we recall some facts and definitions about the Fock space and Wick
quantization. We then present our main results in Section~\ref{sec:results}
in Theorems~\ref{thm:integral-formula} and~\ref{thm:exponential-formula}
and illustrate them by a simple example. Section~\ref{sec:Classical-evolution}
and Section~\ref{sec:Existence-quantum-evolution} are devoted to
the construction and properties of the classical and quantum evolution
associated with a symmetric quadratic Hamiltonian. Section~\ref{sec:integral-formula}
and Section~\ref{sec:exponential-formula} contain the proofs of
our two expansion formulae. For the convenience of the reader we recall
some facts about real-linear symplectomorphisms and symplectic Fourier
transform in the appendices.

\section{\label{sec:Wick-calculus}Wick calculus with polynomial observables}

\subsubsection{Definitions}

We recall some definitions and results about Wick quantization. More
details can be found in~\cite{MR2465733}.

In this paper~$\left(\mathcal{Z},\left\langle \cdot,\cdot\right\rangle \right)$
denotes a separable Hilbert space over~$\mathbb{C}$, the field of
complex numbers. It is also a symplectic space with respect to the
symplectic form~$\sigma\left(z_{1},z_{2}\right)=\Im\left\langle z_{1},z_{2}\right\rangle $.
We use the physicists convention that all the scalar products over
Hilbert spaces are linear with respect to the right variable and antilinear
with respect to the left variable.  We denote by $\mathcal{S}_{m}$
the symmetrization operator on~$\bigotimes^{m}\mathcal{Z}$ (the
completion for the natural Hilbert scalar product of the algebraic
tensor product~$\bigotimes^{m,\, alg}\mathcal{Z}$) defined by \[
\mathcal{S}_{m}\left(z_{1}\otimes\cdots\otimes z_{m}\right)=\frac{1}{m!}\sum_{\sigma\in\mathfrak{S}_{m}}z_{\sigma_{1}}\otimes\cdots\otimes z_{\sigma_{m}}\,,\]
where the~$z_{j}$ are vectors in~$\mathcal{Z}$ and~$\mathfrak{S}_{m}$
denotes the set of the permutations of~$\left\{ 1,\dots,m\right\} $.
We will use the notation~$z_{1}\vee\cdots\vee z_{m}$ for~$\mathcal{S}_{m}\left(z_{1}\otimes\cdots\otimes z_{m}\right)$,
and~$z^{\vee m}$ for~$z\vee\cdots\vee z$ when the~$m$ terms
of this product are equal to~$z$. We call \emph{monomial} of \emph{order}~$\left(p,q\right)\in\mathbb{N}^{2}$
a complex-valued application defined on~$\mathcal{Z}$ of the form
\[
b\left(z\right)=\left\langle z^{\vee q},\tilde{b}z^{\vee p}\right\rangle \,,\]
with~$\widetilde{b}\in\mathcal{L}\left(\bigvee^{p}\mathcal{Z},\bigvee^{q}\mathcal{Z}\right)$
where~$\bigvee^{n}\mathcal{Z}$ (or~$\mathcal{Z}^{\vee n}$) denotes
the Hilbert completion of the~$n$-fold symmetric tensor product,
and for two Banach spaces~$E$ and~$F$, the space of continuous
linear applications from~$E$ to~$F$ is denoted by~$\mathcal{L}\left(E,F\right)$.
We then write~$b\in\mathcal{P}_{p,q}\left(\mathcal{Z}\right)$. The
\emph{total order} of~$b$ is the integer~$m=p+q$. The finite linear
combinations of monomials are called \emph{polynomials}. The set of
all polynomials of this type is denoted by~$\mathcal{P}\left(\mathcal{Z}\right)$.
Subsets of particular interest of~$\mathcal{P}\left(\mathcal{Z}\right)$
are~$\mathcal{P}_{m}\left(\mathcal{Z}\right)$ and~$\mathcal{P}_{\leq m}\left(\mathcal{Z}\right)$,
the finite linear combinations of monomials of total order equal to~$m$
and not greater than~$m$.

The Hilbert space \[
\mathcal{H}:=\bigoplus_{n\in\mathbb{N}}\bigvee^{n}\mathcal{Z}\]
is called the symmetric \emph{Fock space} associated with~$\mathcal{Z}$,
where tensor products and sum completions are made with respect to
the natural Hilbert scalar products inherited from~$\mathcal{Z}$.
We also consider the dense subspace~$\mathcal{H}_{\text{fin}}$ of~$\mathcal{H}$
of states with a finite number of particles\[
\mathcal{H}_{\text{fin}}:=\bigoplus_{n\in\mathbb{N}}^{\text{alg}}\bigvee^{n}\mathcal{Z}\,,\]
where the tensor products are completed but the sum is algebraic.

The \emph{Wick quantization} of a monomial~$b\in\mathcal{P}_{p,q}\left(\mathcal{Z}\right)$
is the operator defined on~$\mathcal{H}_{\text{fin}}$ by its action
on~$\bigvee^{n}\mathcal{Z}$ as an element of~$\mathcal{L}(\bigvee^{n}\mathcal{Z},\bigvee^{n+q-p}\mathcal{Z})$,
\[
\left.b^{Wick}\right|_{\bigvee^{n}\mathcal{Z}}=1_{\left[p,+\infty\right)}\left(n\right)\frac{\sqrt{n!\left(n+q-p\right)!}}{\left(n-p\right)!}\varepsilon^{\frac{p+q}{2}}\left(\tilde{b}\vee I_{\bigvee^{n-p}\mathcal{Z}}\right)\,,\]
where~$I_{X}$ denotes the identity map on the space~$X$ and for~$A_{j}\in\mathcal{L}\left(\mathcal{Z}^{\vee p_{j}},\mathcal{Z}^{\vee q_{j}}\right)$,
$A_{1}\vee A_{2}=\mathcal{S}_{q_{1}+q_{2}}A_{1}\otimes A_{2}\mathcal{S}_{p_{1}+p_{2}}$.
The Wick quantization is extended by linearity to polynomials.

We have a notion of \emph{derivative} of a polynomial, first defined
on the monomials and then extended by linearity. For~$b\in\mathcal{P}_{p,q}\left(\mathcal{Z}\right)$
and for any given~$z\in\mathcal{Z}$, the operator \begin{equation}
\partial_{\bar{z}}^{j}\partial_{z}^{k}b\left(z\right):=\frac{p!}{\left(p-k\right)!}\frac{q!}{\left(q-j\right)!}\left(\left\langle z^{\vee\left(q-j\right)}\right|\vee I_{\bigvee^{j}\mathcal{Z}}\right)\tilde{b}\left(z^{\vee\left(p-k\right)}\vee I_{\bigvee^{k}\mathcal{Z}}\right)\label{eq:derivee_poly_wick}\end{equation}
is an element of~$\mathcal{L}\left(\bigvee^{k}\mathcal{Z},\bigvee^{j}\mathcal{Z}\right)$.
We use the {}``bra'' and {}``ket'' notations of the physicists
for vectors and forms in Hilbert spaces. Then we can define the \emph{Poisson
bracket} of order~$k$ of two polynomials~$b_{1}$,~$b_{2}$, by
\[
\left\{ b_{1},b_{2}\right\} ^{\left(k\right)}=\partial_{z}^{k}b_{1}.\partial_{\bar{z}}^{k}b_{2}-\partial_{z}^{k}b_{2}.\partial_{\bar{z}}^{k}b_{1}\]
since, for any polynomial~$b$, $\partial_{z}^{k}b\left(z\right)$
is a $k$-form (on~$\mathcal{Z}$) and~$\partial_{\bar{z}}^{k}b\left(z\right)$
is a $k$-vector.
\begin{rem}
The product denoted by a dot in the definition of the Poisson bracket
is a~$\mathbb{C}$-bilinear duality-product between $k$-forms and
$k$-vectors. As an example consider the polynomials \[
b_{1}\left(z\right)=\left\langle z^{\vee3},\xi_{1}^{\vee3}\right\rangle \left\langle \eta_{1}^{\vee2},z^{\vee2}\right\rangle \quad\mbox{and}\quad b_{2}\left(z\right)=\left\langle z^{\vee3},\xi_{2}^{\vee3}\right\rangle \left\langle \eta_{2},z\right\rangle \,.\]
The Poisson bracket of order~$2$ of~$b_{1}$ and~$b_{2}$ is \[
\left\{ b_{1},b_{2}\right\} ^{\left(2\right)}\left(z\right)=2\times6\times\left\langle z^{\vee3},\xi_{1}^{\vee3}\right\rangle \left\langle \eta_{1}^{\vee2}\vee z,\xi_{2}^{\vee3}\right\rangle \left\langle \eta_{2},z\right\rangle -0.\]

\end{rem}

\subsubsection{Some examples of Wick quantizations}

Here is a quick review of the notations used for some useful examples
of Wick quantization. A vector of~$\mathcal{Z}$ is denoted by~$\xi$,
$A$ is a bounded operator and~$z$ is the variable of the polynomials.
In the next table, the first column describes the polynomial and the
second the corresponding Wick quantization (as an operator on~$\mathcal{H}_{fin}$).

\[
\begin{array}{ccc}
\left\langle z,Az\right\rangle  & \leftrightarrow & \mbox{d}\Gamma\left(A\right)\\
\left|z\right|^{2} & \leftrightarrow & N\\
\left\langle z,\xi\right\rangle  & \leftrightarrow & a^{*}\left(\xi\right)\\
\left\langle \xi,z\right\rangle  & \leftrightarrow & a\left(\xi\right)\\
\sqrt{2}\Re\left\langle z,\xi\right\rangle  & \leftrightarrow & \Phi\left(\xi\right)\end{array}\]
The operator~$\mbox{d}\Gamma\left(A\right)$ is the usual second
quantization of an operator restricted to~$\mathcal{H}_{fin}$ multiplied
by a factor~$\varepsilon$. If~$A=I_{\mathcal{Z}}$ we obtain~$N$
the usual number operator multiplied by a factor~$\varepsilon$.
The operators~$a$,~$a^{*}$ and~$\Phi$ are the usual annihilation,
creation and field operators of quantum field theory with an additional~$\sqrt{\varepsilon}$
factor. The real and imaginary parts of a complex number~$\zeta$
are denoted by~$\Re\zeta$ and~$\Im\zeta$. The field operators~$\Phi\left(\xi\right)$
are essentially self-adjoint and this enables us to define the ($\varepsilon$-dependent)
Weyl operators \[
W\left(\xi\right)=e^{i\Phi\left(\xi\right)}\,.\]

\subsubsection{Calculus}

Here are some calculation rules for Wick quantizations of polynomials
in~$\mathcal{P}\left(\mathcal{Z}\right)$. The proofs can be found
in~\cite{MR2465733}.
\begin{prop}
\label{pro:Wick-calculus}For every polynomial~$b\in\mathcal{P}\left(\mathcal{Z}\right)$,
\begin{itemize}
\item $b_{1}^{Wick}b_{2}^{Wick}=\left(\sum_{k=0}^{\min\left\{ p_{1},q_{2}\right\} }\frac{\varepsilon^{k}}{k!}\partial_{z}^{k}b_{1}.\partial_{\bar{z}}^{k}b_{2}\right)^{Wick}$
in~$\mathcal{H}_{fin}$ for any~$b_{i}\in\mathcal{P}_{p_{i},q_{i}}\left(\mathcal{Z}\right)$,
\item $b^{Wick}$ is closable and the domain of the closure contains \[
\mathcal{H}_{0}=\mbox{Vect}\left\{ W\left(z\right)\varphi,\varphi\in\mathcal{H}_{fin},\, z\in\mathcal{Z}\right\} \,,\]
(we still denote by $b^{Wick}$ the closure of $b^{Wick}$),
\item $\left(b^{Wick}\right)^{*}=\bar{b}^{Wick}$ on $\mathcal{H}_{fin}$
(where the bar denotes the usual conjugation on complex numbers),
\item for any $z_{0}$ in $\mathcal{Z}$, $W\left(\frac{\sqrt{2}}{i\varepsilon}z_{0}\right)^{*}b^{Wick}W\left(\frac{\sqrt{2}}{i\varepsilon}z_{0}\right)=\left(b\left(z_{0}+z\right)\right)^{Wick}$
holds on~$\mathcal{H}_{0}$ where~$b\left(z_{0}+\cdot\right)\in\mathcal{P}\left(\mathcal{Z}\right)$.
\end{itemize}
\end{prop}

\section{Main results and a simple example\label{sec:results}}

Our two hypotheses are:
\begin{description}
\item [{H1}] Let~$\left(\alpha_{t}\right)_{t\in\mathbb{R}}$ be a one
parameter family of self-adjoint operators on~$\mathcal{Z}$ defining
a strongly continuous dynamical system~$u_{\alpha}(t,s)$.
\item [{H1'}] Assume~\textbf{H1} and additionally that the dynamical system
preserves a dense set~$D$ such that, for any~$\psi\in D$, $u_{\alpha}\left(\cdot,\cdot\right)\psi$
belongs to~$\mathcal{C}^{1}\left(\mathbb{R}^{2},\mathcal{Z}\right)\cap\mathcal{C}^{0}\left(\mathbb{R}^{2},D\right)$.
\item [{H2}] Let~$\beta$ be in~$\mathcal{C}^{0}\left(\mathbb{R};\mathcal{Z}^{\vee2}\right)$,
($\beta_{t}$ defines a $\mathbb{C}$-antilinear Hilbert-Schmidt operator
by~$z\mapsto\left(I_{\mathcal{Z}}\vee\left\langle z\right|\right)\beta_{t}$).
\end{description}
With~\textbf{H1'} and~\textbf{H2}, the \emph{classical flow} associated
with~$Q_{t}\left(z\right)=\left\langle z,\alpha_{t}z\right\rangle +\Im\left\langle \beta_{t},z^{\vee2}\right\rangle $
of quadratic polynomials is the solution~$\varphi\left(t,s\right)$
to the equation\begin{equation}
\left\{ \begin{array}{rcl}
i\partial_{t}\varphi\left(t,0\right)\left[z\right] & = & \partial_{\bar{z}}Q_{t}\left(\varphi\left(t,0\right)\left[z\right]\right)\\
\varphi\left(0,0\right) & = & I_{\mathcal{Z}}\end{array}\right.\label{eq:classical-flow}\end{equation}
where~$\partial_{\bar{z}}Q_{t}\left(z\right)=\alpha z+i\left(I_{\mathcal{Z}}\vee\left\langle z\right|\beta\right)$,
written in a weak sense.

Although things are better visualized by writing a differential equation,
the hypotheses~\textbf{H1} and~\textbf{H2} suffice to define the
dynamical system~$\varphi\left(t,s\right)$. Details about this point
are given in Section~\ref{sec:Classical-evolution}. Actually~$\varphi\left(t,s\right)$
is a family of symplectomorphisms of~$\left(\mathcal{Z},\sigma\right)$
which are naturally decomposed into their $\mathbb{C}$-linear and
$\mathbb{C}$-antilinear parts: \[
\varphi=L+A\,,\quad L\in\mathcal{L}\left(\mathcal{Z}\right)\,,\quad AA^{*}\in\mathcal{L}_{1}\left(\mathcal{Z}\right)\,.\]
See Appendix~\ref{sec:symplectic-transformations} for more details
about symplectomorphisms and this decomposition.

Similarly, the \emph{quantum flow} associated with~$Q_{t}$ is the
solution~$U\left(t,s\right)$ of \begin{equation}
\left\{ \begin{array}{rcl}
i\varepsilon\partial_{t}U\left(t,0\right) & = & Q_{t}^{Wick}U\left(t,0\right)\\
U\left(0,0\right) & = & I_{\mathcal{H}}\end{array}\right.\,.\label{eq:quantum-flow}\end{equation}
The precise meaning of the solutions to this equation is specified
in Section~\ref{sec:Existence-quantum-evolution}. 

We are ready to state our two main results dealing with the evolution
of a Wick observable~$b^{Wick}$, $b\in\mathcal{P}\left(\mathcal{Z}\right)$,
under the quantum flow, that is to say the quantity~$U\left(0,t\right)b^{Wick}U\left(t,0\right)$.
(We use the usual notation~$\left\langle N\right\rangle =\sqrt{N^{2}+1}$.)
\begin{thm}
\label{thm:integral-formula}Assume~\textbf{H1} and~\textbf{H2}.
Let~$b\in\mathcal{P}_{\leq m}\left(\mathcal{Z}\right)$ be a polynomial.
Then, for any time~$t\geq0$, the formula\begin{equation}
U\left(0,t\right)b^{Wick}U\left(t,0\right)=\left(b^{\left(0\right),t}\right)^{Wick}+\sum_{k=1}^{\left\lfloor m/2\right\rfloor }\left(\frac{\varepsilon}{2}\right)^{k}\int_{\Delta_{t}^{k}}\left(b^{\left(k\right)t,\bar{s}^{k}}\right)^{Wick}d\bar{s}^{k}\label{eq:integral-formula}\end{equation}
holds as an equality of continuous operators from~$\mathcal{D}\left(\left\langle N\right\rangle ^{m/2}\right)$
to~$\mathcal{H}$, where
\begin{itemize}
\item $\bar{s}^{k}=\left(s_{1},\dots,s_{k}\right)\in\mathbb{R}_{+}^{k}$
and~$\Delta_{t}^{k}=\left\{ \bar{s}^{k}\in\mathbb{R}_{+}^{k},\,\sum_{j=1}^{k}s_{j}\leq t\right\} $,
\item the polynomials~$b^{\left(k\right)t,\bar{s}^{k}}$ are defined recursively
by\[
\left\{ \begin{array}{rcl}
b^{\left(0\right)t}\left(z\right) & = & b\left(\varphi\left(t,0\right)z\right)\\
b^{\left(k+1\right)t,\bar{s}^{k+1}} & = & \lambda^{s_{k+1}}b^{\left(k\right)t,\bar{s}^{k}}\end{array}\right.\,,\]
with~$\lambda^{s}c=-i\left\{ c\circ\varphi\left(0,s\right),Q_{s}\right\} ^{\left(2\right)}\circ\varphi\left(s,0\right)$
 for any polynomial~c.
\end{itemize}
\end{thm}

\begin{thm}
\label{thm:exponential-formula}Assume~\textbf{H1} and~\textbf{H2}.
Let~$m\geq2$ and~$b\in\mathcal{P}_{\leq m}\left(\mathcal{Z}\right)$
a polynomial. Then introducing
\begin{itemize}
\item the vector~$v_{t}\in\bigotimes^{2}\mathcal{Z}$ such that for all~$z_{1},\, z_{2}\in\mathcal{Z}$,
\[
\left\langle z_{1}\otimes z_{2},v_{t}\right\rangle =\left\langle z_{1},L^{*}\left(t,0\right)A\left(t,0\right)z_{2}\right\rangle \,,\]

\item the operator on~$\mathcal{P}\left(\mathcal{Z}\right)$ \[
\Lambda^{t}c\left(z\right)=\Tr\left[-2A^{*}\left(t,0\right)A\left(t,0\right)\partial_{\bar{z}}\partial_{z}c\left(z\right)\right]+\left\langle v_{t}\right|.\,\partial_{\bar{z}}^{2}c\left(z\right)+\partial_{z}^{2}c\left(z\right).\left|v_{t}\right\rangle \,,\]

\end{itemize}
the formula\begin{equation}
U\left(0,t\right)b^{Wick}U\left(t,0\right)=\left(e^{\frac{\varepsilon}{2}\Lambda^{t}}\left(b\circ\varphi\left(t,0\right)\right)\right)^{Wick}\label{eq:exponential-formula}\end{equation}
holds as an equality of continuous operators from~$\mathcal{D}\left(\left\langle N\right\rangle ^{m/2}\right)$
to~$\mathcal{H}$.\end{thm}
\begin{rem}
The derivative~$\partial_{\bar{z}}\partial_{z}c\left(z\right)$ is
in~$\mathcal{L}\left(\mathcal{Z}\right)$ and~$\mbox{Tr}$ denotes
the trace on the subset of trace class operators of~$\mathcal{L}\left(\mathcal{Z}\right)$.
\end{rem}

\begin{rem}
For~$m\geq2$ the operators~$\lambda^{t}$ and~$\Lambda^{t}$ send~$\mathcal{P}_{m}\left(\mathcal{Z}\right)$
into~$\mathcal{P}_{m-2}\left(\mathcal{Z}\right)$.
\end{rem}

\begin{rem}
 The exponential is intended in the sense\[
e^{\frac{\varepsilon}{2}\Lambda^{t}}b=\sum_{k=0}^{\left\lfloor m/2\right\rfloor }\frac{1}{k!}\left(\frac{\varepsilon}{2}\Lambda^{t}\right)^{k}b\]
for a polynomial~$b$ in~$\mathcal{P}_{\leq m}\left(\mathcal{Z}\right)$.\end{rem}
\begin{example}
To give an idea of the behavior of these formulae we apply them in
the simplest (non trivial) possible situation, with~$\mathcal{Z}=\mathbb{C}$
and~$Q_{t}\left(z\right)=\Im\left(z^{2}\right)$. As~$Q_{t}$ is
time-independent the classical evolution equation is autonomous and
thus we can write~$\varphi\left(t,s\right)=\varphi\left(t-s\right)$
and~$i\partial_{t}\varphi\left(t\right)z=\partial_{\bar{z}}Q\left(\varphi\left(t\right)z\right)=i\overline{\varphi\left(t\right)z}$.
The solution is~$\varphi\left(t\right)z=z\cosh t+\bar{z}\sinh t$.
We can then compute both \[
\int_{0}^{t}b^{\left(1\right)t,s}ds\qquad\mbox{and}\qquad\Lambda^{t}\left(b\circ\varphi\left(t\right)\right)\,.\]
The first one is easily computed as~$\partial_{z}^{2}Q\left(z\right)=-i$,
$\partial_{\bar{z}}^{2}Q\left(z\right)=i$ and, with~$c=b\circ\varphi\left(t\right)$,
\begin{eqnarray*}
-i\left\{ c\circ\varphi\left(-s\right),Q\left(z\right)\right\} ^{\left(2\right)} & = & \left(\partial_{z}^{2}+\partial_{\bar{z}}^{2}\right)\left(c\circ\varphi\left(-s\right)\right)\\
 & = & \left[\cosh\left(-2s\right)\left(\partial_{z}^{2}+\partial_{\bar{z}}^{2}\right)c\right.\\
 &  & \left.+2\sinh\left(-2s\right)\partial_{\bar{z}}\partial_{z}c\right]\circ\varphi\left(-s\right)\end{eqnarray*}
and thus\begin{eqnarray*}
\int_{0}^{t}b^{\left(1\right)t,s}ds & = & \int_{0}^{t}\left(\cosh\left(-2s\right)\left(\partial_{z}^{2}+\partial_{\bar{z}}^{2}\right)+2\sinh\left(-2s\right)\partial_{\bar{z}}\partial_{z}\right)ds\left(b\circ\varphi\left(t\right)\right)\\
 & = & \left(\frac{1}{2}\sinh\left(2t\right)\left(\partial_{z}^{2}+\partial_{\bar{z}}^{2}\right)+\left(1-\cosh\left(2t\right)\right)\partial_{\bar{z}}\partial_{z}\right)\left(b\circ\varphi\left(t\right)\right)\,.\end{eqnarray*}
Now we compute the second one. Since~$L\left(t,0\right)z=L^{*}\left(t,0\right)z=z\cosh t$
and~$A\left(t,0\right)z=A^{*}\left(t,0\right)z=\bar{z}\sinh t$,
we get~$v_{t}=\cosh t\sinh t$ and then obtain directly\[
\Lambda^{t}=\left(1-\cosh\left(2t\right)\right)\partial_{\bar{z}}\partial_{z}+\frac{1}{2}\sinh\left(2t\right)\left(\partial_{z}^{2}+\partial_{\bar{z}}^{2}\right)\,.\]
We thus obtain the same result with the two computations for the
term of order~1 in~$\varepsilon$.

Then we can show that \[
\int_{\Delta_{t}^{k}}b^{\left(k\right)t,\bar{s}^{k}}d\bar{s}^{k}=\frac{1}{k!}\left(\Lambda^{t}\right)^{k}\left(b\circ\varphi\left(t\right)\right)\]
since\begin{multline*}
\int_{\Delta_{t}^{k}}\prod_{j=1}^{k}\left(2\sinh\left(-2s_{j}\right)\partial_{\bar{z}}\partial_{z}+\cosh\left(-2s_{j}\right)\left(\partial_{z}^{2}+\partial_{\bar{z}}^{2}\right)\right)d\bar{s}^{k}\\
=\frac{1}{k!}\left(\left(1-\cosh\left(-2t\right)\right)\partial_{\bar{z}}\partial_{z}-\frac{1}{2}\sinh\left(-2t\right)\left(\partial_{z}^{2}+\partial_{\bar{z}}^{2}\right)\right)^{k}\end{multline*}
because \[
\frac{d}{ds}\left[\left(1-\cosh\left(-2s\right)\right)\partial_{\bar{z}}\partial_{z}-\frac{1}{2}\sinh\left(-2s\right)\left(\partial_{z}^{2}+\partial_{\bar{z}}^{2}\right)\right]=2\sinh\left(-2s\right)\partial_{\bar{z}}\partial_{z}+\cosh\left(-2s\right)\left(\partial_{z}^{2}+\partial_{\bar{z}}^{2}\right)\,.\]
\end{example}
\begin{rem}
Since these two formulae will be proven independently and the identification
of each term of order~$k$ in~$\varepsilon$ in the expansion of
the symbol is clear, we carry out a computation only on the formal
level for the convenience of the reader to show the link between the
two formulae in the general case. 

We show (formally) that \[
\frac{d}{ds}\Lambda^{s}=\lambda^{s}\,.\]
Then it is simple to show that \[
\int_{\bar{s}^{k}\in\Delta_{t}^{k}}\lambda^{s_{k}}\lambda^{s_{k-1}}\cdots\lambda^{s_{1}}d\bar{s}^{k}=\frac{1}{k!}\left(\Lambda^{t}\right)^{k}\]
as operators on~$\mathcal{P}\left(\mathcal{Z}\right)$ once the case~$k=2$
is understood:\begin{eqnarray*}
2\int_{\bar{s}^{2}\in\Delta_{t}^{2}}\lambda^{s_{2}}\lambda^{s_{1}}d\bar{s}^{2} & = & \int_{0}^{t}\int_{0}^{s_{1}}\lambda^{s_{2}}\lambda^{s_{1}}ds_{2}ds_{1}+\int_{0}^{t}\int_{0}^{s_{2}}\lambda^{s_{2}}\lambda^{s_{1}}ds_{1}ds_{2}\\
 & = & \int_{0}^{t}\Lambda^{s_{1}}\lambda^{s_{1}}ds_{1}+\int_{0}^{t}\lambda^{s_{2}}\Lambda^{s_{2}}ds_{2}\\
 & = & \left(\Lambda^{t}\right)^{2}\,.\end{eqnarray*}
In this computation we have used that~$\Lambda^{0}=0$ as~$A\left(0,0\right)=0$.

We first give~$\lambda^{s}$ in a more explicit way. As~$\partial_{\bar{z}}^{2}Q=i\left|\beta\right\rangle $
and~$\partial_{z}^{2}Q=-i\left\langle \beta\right|$ we first get\[
\lambda c=\left[\partial_{z}^{2}\left(c\circ\varphi^{-1}\right).\left|\beta\right\rangle +\left\langle \beta\right|.\partial_{\bar{z}}^{2}\left(c\circ\varphi^{-1}\right)\right]\circ\varphi\]
with~$\varphi=\varphi\left(t,0\right)$ and omitting the time dependence
everywhere. Then with~$\varphi=L+A$ (and thus~$\varphi^{-1}=L^{*}-A^{*}$)
and~$\left\langle z_{1},Az_{2}\right\rangle =\left\langle z_{1}\otimes z_{2},w_{A}\right\rangle $
we obtain\begin{eqnarray*}
\lambda c\left(z\right) & = & \partial_{z}^{2}c\left(z\right).\left|\left(L^{*\vee2}+A^{*\vee2}\right)\beta\right\rangle +\left\langle \left(L^{*\vee2}+A^{*\vee2}\right)\beta\right|.\partial_{\bar{z}}^{2}c\left(z\right)\\
 &  & -2\left(\left\langle \left(I_{\mathcal{Z}}\otimes\partial_{\bar{z}}\partial_{z}c\left(z\right)^{*}L^{*}\right)\beta,w_{A}\right\rangle +\left\langle w_{A},\left(I_{\mathcal{Z}}\otimes\partial_{\bar{z}}\partial_{z}c\left(z\right)L^{*}\right)\beta\right\rangle \right)\,.\end{eqnarray*}

Then we compute~$\frac{d}{ds}\Lambda^{s}$ in several parts. The
linear and antilinear parts of the equation $i\partial_{s}\varphi\left(s,0\right)z=\partial_{\bar{z}}Q_{s}\left(\varphi\left(s,0\right)z\right)$
give \begin{eqnarray*}
\partial_{s}Lz & = & -i\alpha Lz+\left(\left\langle Az\right|\vee I_{\mathcal{Z}}\right)\left|\beta\right\rangle \\
\partial_{s}Az & = & -i\alpha Az+\left(\left\langle Lz\right|\vee I_{\mathcal{Z}}\right)\left|\beta\right\rangle \,.\end{eqnarray*}
We now show that~$\partial_{s}v_{s}=\left|\left(L^{*\vee2}+A^{*\vee2}\right)\beta\right\rangle $,\begin{eqnarray*}
\partial_{s}\left\langle z_{1}\otimes z_{2},v_{s}\right\rangle  & = & \partial_{s}\left\langle Lz_{1},Az_{2}\right\rangle \\
 & = & \left\langle -i\alpha Lz_{1},Az_{2}\right\rangle +\left\langle \beta,Az_{2}\vee Az_{1}\right\rangle \\
 &  & +\left\langle Lz_{1},-i\alpha Az_{2}\right\rangle +\left(\left\langle Lz_{2}\right|\vee\left\langle Lz_{1}\right|\right)\left|\beta\right\rangle \\
 & = & \left\langle \beta,\left(A\vee A\right)\left(z_{1}\vee z_{2}\right)\right\rangle +\left\langle \left(L\vee L\right)\left(z_{1}\vee z_{2}\right),\beta\right\rangle \\
 & = & \left\langle z_{1}\vee z_{2},\left(L^{*\vee^{2}}+A^{*\vee^{2}}\right)\beta\right\rangle \,.\end{eqnarray*}
And thus~$\partial_{s}\left(\partial_{z}^{2}.\left|v\right\rangle +\left\langle v\right|.\partial_{\bar{z}}^{2}\right)=\partial_{z}^{2}.\left|\left(L^{*\vee2}+A^{*\vee2}\right)\beta\right\rangle +\left\langle \left(L^{*\vee2}+A^{*\vee2}\right)\beta\right|.\partial_{\bar{z}}^{2}$.

We then show that \[
\partial_{s}\mbox{Tr}\left[A^{*}A\partial_{\bar{z}}\partial_{z}c\left(z\right)\right]=\left\langle \beta,\left(I_{\mathcal{Z}}\otimes L\partial_{\bar{z}}\partial_{z}c\left(z\right)\right)w_{A}\right\rangle +\left\langle w_{A},\left(I_{\mathcal{Z}}\otimes\partial_{\bar{z}}\partial_{z}c\left(z\right)L^{*}\right)\beta\right\rangle \,.\]
We first observe that~$\mbox{Tr}\left[A^{*}A\partial_{\bar{z}}\partial_{z}c\left(z\right)\right]=\left\langle w_{A},\left(I_{\mathcal{Z}}\otimes\partial_{\bar{z}}\partial_{z}c\left(z\right)\right)w_{A}\right\rangle $.
A simple calculation using~$\partial_{s}Az=-i\alpha Az+\left(\left\langle Lz\right|\vee I_{\mathcal{Z}}\right)\left|\beta\right\rangle $
shows that~$\partial_{s}w_{A}=\left(-i\alpha\otimes I_{\mathcal{Z}}\right)w_{A}+\left(I_{\mathcal{Z}}\otimes L^{*}\right)\beta$
and this immediately gives the result.
\end{rem}

\section{Classical evolution of a Wick polynomial under a quadratic evolution\label{sec:Classical-evolution}}

The adjoint of a $\mathbb{C}$-antilinear operator is defined in Appendix~\ref{sec:symplectic-transformations}.
\begin{defn}
A $\mathbb{C}$-antilinear operator~$A$ on~$\mathcal{Z}$ is said
of \emph{Hilbert-Schmidt class} if~$\left\Vert A\right\Vert _{\mathcal{L}_{2}^{a}\left(\mathcal{Z}\right)}:=\left\Vert AA^{*}\right\Vert _{\mathcal{L}_{1}\left(\mathcal{Z}\right)}^{1/2}$
is finite, where~$\left\Vert \cdot\right\Vert _{\mathcal{L}_{1}\left(\mathcal{Z}\right)}$
is the usual trace norm for $\mathbb{C}$-linear operators. The set
of Hilbert-Schmidt antilinear operators is denoted by~$\mathcal{L}_{2}^{a}\left(\mathcal{Z}\right)$.

Let~$\mathcal{X}\left(\mathcal{Z}\right)=\mathcal{L}\left(\mathcal{Z}\right)+\mathcal{L}_{2}^{a}\left(\mathcal{Z}\right)$
with norm~\[
\left\Vert T\right\Vert _{\mathcal{X}\left(\mathcal{Z}\right)}=\left\Vert L\right\Vert _{\mathcal{L}\left(\mathcal{Z}\right)}+\left\Vert A\right\Vert _{\mathcal{L}_{2}^{a}\left(\mathcal{Z}\right)}\]
 for~$T=L+A$, where~$L$ and~$A$ are respectively $\mathbb{C}$-linear
and $\mathbb{C}$-antilinear. The space~$\mathcal{X}\left(\mathcal{Z}\right)$
is a Banach algebra.\end{defn}
\begin{rem}
The norm~$\left\Vert T\right\Vert _{\mathcal{X}\left(\mathcal{Z}\right)}$
is well defined as the decomposition~$T=L+A$ is unique ($L=\frac{1}{2}\left(T-iTi\right)$
and~$A=\frac{1}{2}\left(T+iTi\right)$).
\end{rem}

\subsubsection{Construction of the classical flow without the~$\alpha$ term}

Let~$\beta\in\mathcal{C}^{0}\left(\mathbb{R};\mathcal{Z}^{\vee2}\right)$
and~$Q_{t}=\Im\left\langle \beta_{t},z^{\vee2}\right\rangle $. 
Observe that~$\partial_{\bar{z}}Q\left(t\right)\left(z\right)=i\left(I_{\mathcal{Z}}\vee\left\langle z\right|\right)\beta_{t}$
and so~$\left(\partial_{\bar{z}}Q_{t}\right)_{t}$ is a continous
one parameter family of~$\mathcal{X}\left(\mathcal{Z}\right)$, so
that the theory of ordinary differential equations in Banach algebras
(see for example~\cite{MR0249698}) asserts that there exists a unique
two parameters family~$\varphi\left(t_{2},t_{1}\right)$ of elements
of~$\mathcal{X}\left(\mathcal{Z}\right)$ such that\[
\left\{ \begin{array}{rcl}
i\partial_{t}\varphi\left(t,0\right) & = & \partial_{\bar{z}}Q_{t}\;\varphi\left(t,0\right)\\
\varphi\left(0,0\right) & = & I_{\mathcal{Z}}\end{array}\right.\,,\]
with~$\varphi$ of~$\mathcal{C}^{1}$ class in both parameters such
that for all~$r$,~$s$ and~$t$, \[
\varphi\left(t,s\right)\varphi\left(s,r\right)=\varphi\left(t,r\right)\,.\]

The classical flow~$\varphi\left(t,s\right)$ is a symplectomorphism
with respect to the symplectic form~$\sigma\left(z_{1},z_{2}\right)=\Im\left\langle z_{1},z_{2}\right\rangle $.
It can be checked deriving\[
\sigma\left(\varphi\left(t,s\right)z_{1},\varphi\left(t,s\right)z_{2}\right)\]
 with respect to~$t$.

\subsubsection{The strongly continuous dynamical system associated with~$\left(\alpha_{t}\right)$}

We first state a proposition which is a direct consequence of Theorem~X.70
in~\cite{MR0493420} in the unitary case. This proposition provides
a set of assumptions ensuring the existence of a strongly continuous
dynamical system associated with a family~$\left(\alpha_{t}\right)_{t}$
of self-adjoint operators. Other more general situations can be considered
as in~\cite{MR0279626,MR0161185} for example.
\begin{prop}
\label{thm:reed-simon-system-dynamique-fortement-continu}Let~$\left(\alpha_{t}\right)_{t\in\mathbb{R}}$
be a family of self-adjoint operators on the Hilbert space~$\mathcal{Z}$
satisfying the following conditions.
\begin{enumerate}
\item The~$\alpha_{t}$ have a common domain~$D$ (from which it follows
by the closed graph theorem that~$c\left(t,s\right)=\left(\alpha_{t}-i\right)\left(\alpha_{s}-i\right)^{-1}$
is bounded).
\item For each~$z\in\mathcal{Z}$, $\left(t-s\right)^{-1}c\left(t,s\right)z$
is uniformly strongly continuous and uniformly bounded in~$s$ and~$t$
for~$t\neq s$ lying in any fixed compact interval.
\item For each~$z\in\mathcal{Z}$, $c\left(t\right)z=\lim_{s\nearrow t}\left(t-s\right)^{-1}c\left(t,s\right)z$
exists uniformly for~$t$ in each compact interval and~$c\left(t\right)$
is bounded and strongly continuous in~$t$.
\end{enumerate}
The approximate propagator~$u_{k}$ is defined by $u_{k}\left(t,s\right)=\exp(-\left(t-s\right)i\alpha_{\frac{j-1}{k}})$
if~$\frac{j-1}{k}\leq s\leq t\leq\frac{j}{k}$ and~$u_{k}\left(t,r\right)=u_{k}\left(t,s\right)u_{k}\left(s,r\right)$.

Then for all~$s$, $t$ in a compact interval and any~$z\in\mathcal{Z}$,\[
u\left(t,s\right)z=\lim_{k\to+\infty}u_{k}\left(t,s\right)z\]
exists uniformly in~$s$ and~$t$. Further, if~$z\in D$, then~$u\left(t,s\right)z$
is in~$D$ for all~$s,\, t$ and satisfies\[
\left\{ \begin{array}{rcl}
i\frac{d}{dt}u\left(t,s\right)z & = & \alpha_{t}u\left(t,s\right)z\\
u\left(s,s\right)z & = & z\end{array}\right.\,.\]

\end{prop}

\subsubsection{Construction of the classical flow with the~$\alpha$ term}

Assume~\textbf{H1} and~\textbf{H2}. Let~$\hat{\varphi}$ be the
solution of\[
\left\{ \begin{array}{rcl}
i\partial_{t}\hat{\varphi}\left(t,0\right) & = & \partial_{\bar{z}}\hat{Q}_{t}\;\hat{\varphi}\left(t,0\right)\\
\hat{\varphi}\left(0,0\right) & = & I_{\mathcal{Z}}\end{array}\right.\,,\]
with~$ $$\hat{Q}_{t}\left(z\right)=\Im\left\langle \hat{\beta}_{t},z^{\vee2}\right\rangle $,
$\hat{\beta}_{t}=u_{\alpha}\left(t,0\right)^{*\vee2}\beta_{t}$. What
we call here the solution of\begin{equation}
\left\{ \begin{array}{rcl}
i\partial_{t}\varphi\left(t,0\right) & = & \partial_{\bar{z}}Q_{t}\;\varphi\left(t,0\right)\\
\varphi\left(0,0\right) & = & I_{\mathcal{Z}}\end{array}\right.\,,\label{eq:def-phi}\end{equation}
with~$Q_{t}=\left\langle z,\alpha_{t}z\right\rangle +\Im\left\langle \beta_{t},z^{\vee2}\right\rangle $
is\[
\varphi\left(t,0\right)=u_{\alpha}\left(t,0\right)\circ\hat{\varphi}\left(t,0\right)\,.\]

Depending on the assumptions on~$\left(\alpha_{t}\right)$ it will
be possible to precise if~$\varphi$ sloves Equation~(\ref{eq:def-phi})
in a usual sense (strongly, weakly, on some dense subset...).

With the particular set of assumptions of Theorem~\ref{thm:reed-simon-system-dynamique-fortement-continu}
we get that for all~$z_{1}\in D$ and~$z_{2}\in\mathcal{Z}$,\[
\left\{ \begin{array}{rcl}
i\partial_{t}\left\langle z_{1},\varphi\left(t,0\right)z_{2}\right\rangle  & = & \left\langle \alpha z_{1},\varphi\left(t,0\right)z_{2}\right\rangle +i\left\langle z_{1}\vee\varphi\left(t,0\right)z_{2},\beta\right\rangle \\
\varphi\left(0,0\right) & = & I_{\mathcal{Z}}\end{array}\right.\,.\]

\subsubsection{Composition of a Wick polynomial with the classical evolution}

The composition of a polynomial with the classical flow defines a
time-dependent polynomial.
\begin{defn}
We define a norm on~$\mathcal{P}\left(\mathcal{Z}\right)$ by \[
\left\Vert b\right\Vert _{\mathcal{P}\left(\mathcal{Z}\right)}=\sum_{p,\, q}\left\Vert b_{p,q}\right\Vert _{q\leftarrow p}\]
where~$b=\sum_{p,\, q}b_{p,q}$ is a polynomial with~$b_{p,q}\in\mathcal{P}_{p,q}\left(\mathcal{Z}\right)$
and~$\left\Vert b_{p,q}\right\Vert _{q\leftarrow p}$ is a shorthand
for $\|\tilde{b}_{p,q}\|_{\mathcal{L}\left(\bigvee^{p}\mathcal{Z},\bigvee^{q}\mathcal{Z}\right)}$.
For a polynomial~$b$ in~$\mathcal{P}_{m}\left(\mathcal{Z}\right)$,
we will sometimes write~$\left\Vert b\right\Vert _{\mathcal{P}_{m}\left(\mathcal{Z}\right)}$.
\end{defn}

\begin{prop}
\label{pro:estimate_comp_classical_flow}Let~$b\in\mathcal{P}_{m}\left(\mathcal{Z}\right)$
be a polynomial, and~$\varphi\in\mathcal{X}\left(\mathcal{Z}\right)$.
Then~$b\circ\varphi\in\mathcal{P}_{m}\left(\mathcal{Z}\right)$ and
we have the estimate\[
\left\Vert b\circ\varphi\right\Vert _{\mathcal{P}_{m}\left(\mathcal{Z}\right)}\leq\left\Vert \varphi\right\Vert _{\mathcal{X}\left(\mathcal{Z}\right)}^{m}\left\Vert b\right\Vert _{\mathcal{P}_{m}\left(\mathcal{Z}\right)}\,.\]
\end{prop}
\begin{proof}
The proof is essentially the same as in Proposition~2.12 of~\cite{MR2465733}.
\end{proof}
%
{}

\section{Quantum evolution of a Wick polynomial\label{sec:Existence-quantum-evolution}}

\subsubsection{Without the~$\alpha$ term}

\begin{defn}
Let~$\beta\in\mathcal{C}^{0}\left(\mathbb{R};\mathcal{Z}^{\vee2}\right)$
and~$Q_{t}\left(z\right)=\Im\left\langle \beta_{t},z^{\vee2}\right\rangle $.
A family~$U\left(t,s\right)$ of unitary operators on~$\mathcal{H}$
defined for~$s,\, t$ real is a \emph{solution} of \begin{equation}
\left\{ \begin{array}{rcl}
i\partial_{t}U\left(t,0\right) & = & \frac{Q_{t}^{Wick}}{\varepsilon}U\left(t,0\right)\\
U\left(0,0\right) & = & I_{\mathcal{H}}\end{array}\right.\label{eq:quantum_flow}\end{equation}
if
\begin{enumerate}
\item $U\left(t,s\right)$ is strongly continuous in $\mathcal{H}$ with
respect to~$s$,~$t$ with~$U\left(s,s\right)=I$,
\item $U\left(t,r\right)=U\left(t,s\right)U\left(s,r\right)$, $r\leq s\leq t$,
\item $i\frac{d}{dt}U\left(t,s\right)y$ exists for almost every~$t$
(depending on~$s$) and is equal to\ $Q_{t}^{Wick}U\left(t,s\right)y$,
\item $i\varepsilon\frac{d}{ds}U\left(t,s\right)y=-U\left(t,s\right)Q_{s}^{Wick}y$,
$y\in\mathcal{D}\left(N+1\right)$, $0\leq s\leq t$.
\end{enumerate}
\end{defn}
This definition is made to fit the general framework of Theorems~4.1
and~5.1 of~\cite{MR0279626}. More precisely we may check the following
theorem.
\begin{thm}
\label{thm:existence-quantum-flow}Let~$\beta\in\mathcal{C}^{0}\left(\mathbb{R};\mathcal{Z}^{\vee2}\right)$
and~$Q_{t}\left(z\right)=\Im\left\langle \beta_{t},z^{\vee2}\right\rangle $.

Then the quantum flow equation~(\ref{eq:quantum_flow}) associated
to the family~$\frac{1}{\varepsilon}Q_{t}$ has a unique solution.
This solution preserves the sets~$\mathcal{D}(\left\langle N\right\rangle ^{k/2})$
for~$k\geq2$.
\end{thm}
To establish this theorem we will use the following estimates.
\begin{lem}
\label{lem:estimation-QWick-Q-quadratique}Let~$\beta\in\mathcal{Z}^{\vee2}$
and~$Q\left(z\right)=\Im\left\langle \beta,z^{\vee2}\right\rangle $.
Then, on~$\mathcal{H}_{fin}$, and for~$k\geq1$, $Q^{Wick}$ satisfies
the estimates\begin{equation}
\left\Vert Q^{Wick}/\varepsilon\Psi\right\Vert \leq\frac{3}{2}\left\Vert \beta\right\Vert _{\mathcal{Z}^{\vee2}}\left\Vert \left(N/\varepsilon+1\right)\Psi\right\Vert \label{eq:estimate-QWick}\end{equation}
and\begin{equation}
\pm i\left[Q^{Wick}/\varepsilon,\left(N/\varepsilon+1\right)^{k}\right]\leq3^{k}\sqrt{2}\left\Vert \beta\right\Vert _{\mathcal{Z}^{\vee2}}\left(N/\varepsilon+1\right)^{k}\,.\label{eq:estimate-[QWick,N+1]}\end{equation}

\end{lem}
The second estimate is in the sense of quadratic forms, for all~$\Psi\in\mathcal{H}_{fin}$,
\begin{multline*}
\pm i\left(\left\langle \frac{1}{\varepsilon}Q^{Wick}\Psi,\left(N/\varepsilon+1\right)^{k}\Psi\right\rangle -\left\langle \left(N/\varepsilon+1\right)^{k}\Psi,\frac{1}{\varepsilon}Q^{Wick}\Psi\right\rangle \right)\\
\leq\frac{3^{k}}{\sqrt{2}}\left\Vert \beta\right\Vert _{\mathcal{Z}^{\vee2}}\left\langle \Psi,\left(N/\varepsilon+1\right)^{k}\Psi\right\rangle \,.\end{multline*}

\begin{proof}
The first estimate is a consequence of~$n+2\leq2\left(n+1\right)$
associated to\[
\frac{2i}{\varepsilon}\left.Q^{Wick}\right|_{\mathcal{Z}^{\vee n}}=\sqrt{n\left(n-1\right)}\left\langle \beta\right|\vee I_{\bigvee^{n-2}\mathcal{Z}}-\sqrt{\left(n+2\right)\left(n+1\right)}\left|\beta\right\rangle \vee I_{\bigvee^{n}\mathcal{Z}}\,.\]

For the second estimate, consider~$\frac{2i}{\varepsilon}\left\langle \Psi,\left[\left(1+N/\varepsilon\right)^{k},Q^{Wick}\right]\Psi\right\rangle $.
The first term of this commutator is \[
\sum_{n}\left(n+1\right)^{k}\left(\sqrt{\left(n+2\right)\left(n+1\right)}\left\langle \Psi^{\left(n\right)}\vee\left\langle \beta\right|,\Psi^{\left(n+2\right)}\right\rangle \right.\left.-\sqrt{n\left(n-1\right)}\left\langle \Psi^{\left(n\right)},\left|\beta\right\rangle \vee\Psi^{\left(n-2\right)}\right\rangle \right)\,.\]
Then we deduce easily the second term and a reindexation gives the
following form for the whole commutator:\begin{multline*}
\sum_{n}\left[\left(n+1\right)^{k}-\left(\left(n+2\right)+1\right)^{k}\right]\sqrt{\left(n+2\right)\left(n+1\right)}\\
\times\left(\left\langle \Psi^{\left(n\right)}\vee\left\langle \beta\right|,\Psi^{\left(n+2\right)}\right\rangle +\left\langle \Psi^{\left(n+2\right)},\left|\beta\right\rangle \vee\Psi^{\left(n\right)}\right\rangle \right)\,.\end{multline*}
Newton's binomial formula and the inequalities~$\sum_{l=0}^{k-1}\binom{k}{l}2^{k-l}\leq3^{k}$
and $\left(n+1\right)^{l}\leq\left(n+1\right)^{k-1}$ yield \[
\left(n+1\right)^{k}-\left(\left(n+2\right)+1\right)^{k}\leq3^{k}\left(n+1\right)^{k-1}\,.\]
Using also~$n+2\leq2\left(n+1\right)$ to control~$\sqrt{\left(n+2\right)\left(n+1\right)}$
we obtain \[
\pm i\left\langle \Psi,\left[\left(1+N/\varepsilon\right)^{k},\frac{Q^{Wick}}{\varepsilon}\right]\Psi\right\rangle \leq\frac{1}{2}\sum_{n}3^{k}\left(n+1\right)^{k-1}\sqrt{2}\left(n+1\right)\left\Vert \Psi^{\left(n\right)}\right\Vert \left\Vert \beta\right\Vert _{\mathcal{Z}^{\vee2}}\left\Vert \Psi^{\left(n+2\right)}\right\Vert \,.\]
Cauchy-Schwarz's inequality gives the claimed estimate.
\end{proof}

\begin{lem}
\label{lem:QWick-autoadjoint}Let~$\beta\in\mathcal{Z}^{\vee2}$
and~$Q\left(z\right)=\Im\left\langle \beta,z^{\vee2}\right\rangle $.
Then~$Q^{Wick}$ is essentially self-adjoint on~$\mathcal{H}_{fin}$
and its closure is essentially self-adjoint on any other core for~$N/\varepsilon+1$.
Inequalities~(\ref{eq:estimate-QWick}) and~(\ref{eq:estimate-[QWick,N+1]})
still hold on~$\mathcal{D}\left(N/\varepsilon+1\right)$.
\end{lem}
We still denote by~$Q^{Wick}$ this self-adjoint extension.
\begin{proof}
We apply the commutators Theorem X.37 of~\cite{MR0493420} with
the estimates of Lemma~\ref{lem:estimation-QWick-Q-quadratique}
for~$k=1$.\end{proof}
\begin{lem}
\label{lem:stability_by_U}If a solution of the quantum flow equation~(\ref{eq:quantum_flow})
exists then it leaves~$\mathcal{Q}(\left(N/\varepsilon+1\right)^{k})=\mathcal{D}(\left(N/\varepsilon+1\right)^{k/2})$
invariant for any integer~$k\geq2$.

In the time-independent case the estimate \[
\left\Vert U\left(t,0\right)\right\Vert _{\mathcal{L}\left(\mathcal{D}\left(\left(N/\varepsilon+1\right)^{k/2}\right)\right)}\leq\exp\left(3^{k}\sqrt{2}\left\Vert \beta\right\Vert \left|t\right|\right)\]
 holds.\end{lem}
\begin{proof}
From Lemma~\ref{lem:QWick-autoadjoint}, for any~$k\geq2$, $\mathcal{D}(\left(N/\varepsilon+1\right)^{k/2})\subset\mathcal{D}(Q^{Wick})$.
We can adapt the proof of Theorem~2 of~\cite{MR0391794} to the
case of the quantization of a continuous one parameter family of quadratic
polynomials with the estimates of Lemma~\ref{lem:estimation-QWick-Q-quadratique}.
\end{proof}

\begin{proof}[Proof of theorem~\ref{thm:existence-quantum-flow}]
We use Theorems~4.1 and 5.1 of~\cite{MR0279626} with the family
of operators~$iQ\left(t\right)^{Wick}/\varepsilon$ (here we directly
consider the self-adjoint extension of~$Q_{t}^{Wick}/\varepsilon$).
We set~$Y=\mathcal{D}(\left(N/\varepsilon+1\right)^{k/2})$.
\begin{enumerate}
\item This family is stable in the sense that~$\|\prod_{j=1}^{k}e^{-is_{j}Q\left(t_{j}\right)^{Wick}/\varepsilon}\|_{\mathcal{L}\left(\mathcal{H}\right)}\leq1$
(we actually have an equality here).
\item The space~$Y$ is admissible for this family in the sense that for
each~$t$, $(iQ_{t}^{Wick}/\varepsilon+\lambda)^{-1}$ leaves~$Y$
invariant and \[
\left\Vert \left(iQ_{t}^{Wick}/\varepsilon+\lambda\right)^{-1}\right\Vert _{\mathcal{L}\left(Y\right)}\leq\left(\lambda-3^{k}\sqrt{2}\left\Vert \beta\right\Vert \right)^{-1}\]
for~$\Re\lambda>3^{k}\sqrt{2}\left\Vert \beta\right\Vert $.

This is true because, as we have seen in Lemma~\ref{lem:stability_by_U},
$(e^{-isQ_{t}^{Wick}/\varepsilon})_{s\in\mathbb{R}}$ leaves~$Y$
invariant and, thanks to the estimate of the same lemma, we can apply
the resolvent formula \[
\left(iQ_{t}^{Wick}/\varepsilon+\lambda\right)^{-1}=\int_{0}^{+\infty}e^{-\lambda s}e^{-isQ_{t}^{Wick}/\varepsilon}ds\]
and obtain the desired estimate.

\item $ $$Y\subset\mathcal{D}\left(Q_{t}^{Wick}/\varepsilon\right)$ so
that~$Q_{t}^{Wick}/\varepsilon\in\mathcal{L}\left(Y,\mathcal{H}\right)$
for each~$t$, and the map~$t\to Q_{t}^{Wick}/\varepsilon\in\mathcal{L}\left(Y,\mathcal{H}\right)$
is continuous.
\item $Y=\mathcal{D}(\left(N/\varepsilon+1\right)^{k/2})$ is reflexive.
\end{enumerate}
Theorems~4.1 and 5.1 of~\cite{MR0279626} thus apply and give the
existence of an evolution operator.

The preservation of the set~$\mathcal{D}(\left(N/\varepsilon+1\right)^{k/2})$
comes from the application of Lemma~\ref{lem:stability_by_U} to
the solution of the time-dependent problem. To conclude it is then
enough to observe that the domains~$\mathcal{D}(\left\langle N\right\rangle ^{k/2})$
and~$\mathcal{D}(\left(N/\varepsilon+1\right)^{k/2})$ are the same
and have equivalent norms.
\end{proof}

\subsubsection{With the~$\alpha$ term}

Assume~\textbf{H1} and~\textbf{H2}. Let~$\hat{U}$ be the solution
of\begin{equation}
\left\{ \begin{array}{rcl}
i\partial_{t}\hat{U}\left(t,0\right) & = & \frac{\hat{Q}_{t}^{Wick}}{\varepsilon}\hat{U}\left(t,0\right)\\
\hat{U}\left(0,0\right) & = & I_{\mathcal{H}}\end{array}\right.\end{equation}
with~$\hat{Q}_{t}\left(z\right)=\Im\left\langle \hat{\beta}_{t},z^{\vee2}\right\rangle $,
$\hat{\beta}_{t}=u_{\alpha}\left(t,0\right)^{*\vee2}\beta_{t}$. What
we call here the solution of\begin{equation}
\left\{ \begin{array}{rcl}
i\partial_{t}U\left(t,0\right) & = & \frac{Q_{t}^{Wick}}{\varepsilon}U\left(t,0\right)\\
U\left(0,0\right) & = & I_{\mathcal{H}}\end{array}\right.\end{equation}
with~$Q_{t}=\left\langle z,\alpha_{t}z\right\rangle +\Im\left\langle \beta_{t},z^{\vee2}\right\rangle $
is\[
U\left(t,0\right)=\Gamma\left(u_{\alpha}\left(t,0\right)\right)\circ\hat{U}\left(t,0\right)\,.\]

\section{Removal of the~$\alpha$ part}
\begin{prop}
Assume~\textbf{H1} and~\textbf{H2}. Suppose Theorems~\ref{thm:integral-formula}
and~\ref{thm:exponential-formula} hold with a null one parameter
family of self-adjoint operators on~$\mathcal{Z}$, and~$\hat{\beta}_{t}=u_{\alpha}\left(t,0\right)^{*\vee2}\beta_{t}$.
We denote with a hat the quantities associated with this solution.
Then Theorems~\ref{thm:integral-formula} and~\ref{thm:exponential-formula}
hold.\end{prop}
\begin{proof}
For Equation~\ref{eq:integral-formula}, we forget during the proof
the~$\left(t,0\right)$ dependency in our notations and write\[
\int_{\Delta_{t}^{0}}b^{\left(0\right)t,\bar{s}^{0}}d\bar{s}^{0}\]
instead of~$b^{\left(0\right),t}$. Then\begin{eqnarray*}
U^{*}b^{Wick}U & = & \hat{U}^{*}\Gamma\left(u_{\alpha}^{*}\right)b^{Wick}\Gamma\left(u\right)\hat{U}\\
 & = & \hat{U}^{*}\left(b\circ u_{\alpha}\right)^{Wick}\hat{U}\\
 & = & \sum_{k=0}^{\left\lfloor \frac{m}{2}\right\rfloor }\left(\frac{\varepsilon}{2}\right)^{k}\int_{\Delta_{t}^{k}}\left(\widehat{b\circ u_{\alpha}}^{\left(k\right)t,\bar{s}^{k}}\right)^{Wick}d\bar{s}^{k}\end{eqnarray*}
where the~$\hat{b}^{\left(k\right)t,\bar{s}^{k}}$ are defined recursively
by\[
\left\{ \begin{array}{rcl}
\hat{b}^{\left(0\right)t}\left(z\right) & = & b\circ\hat{\varphi}\\
\hat{b}^{\left(k+1\right)t,\bar{s}^{k+1}} & = & \hat{\lambda}^{s_{k+1}}\hat{b}^{\left(k\right)t,\bar{s}^{k}}\end{array}\right.\]
with~$\hat{\lambda}^{s}c=-i\left\{ c\circ\hat{\varphi}\left(0,s\right),\hat{Q}_{s}\right\} ^{\left(2\right)}\circ\hat{\varphi}\left(s,0\right)$
 for any polynomial~$c$. Thus it suffices to prove that\[
\widehat{b\circ u_{\alpha}}^{\left(k\right)t,\bar{s}^{k}}=b^{\left(k\right)t,\bar{s}^{k}}\,.\]
This is clear for~$k=0$ as~$u_{\alpha}\circ\hat{\varphi}=\varphi$.
Then we observe that\begin{eqnarray*}
\hat{\lambda}^{s}c & = & -i\left\{ c\circ\hat{\varphi}^{-1},\hat{Q}\right\} ^{\left(2\right)}\circ\hat{\varphi}\\
 & = & -i\left\{ c\circ\varphi^{-1}\circ u_{\alpha},\hat{Q}\right\} ^{\left(2\right)}\circ u_{\alpha}^{-1}\circ\varphi\\
 & = & -i\left\{ c\circ\varphi^{-1},Q\right\} ^{\left(2\right)}\circ\varphi\end{eqnarray*}
where we used that~$\partial_{z}^{2}\left\langle z,\alpha z\right\rangle =0$,
$\partial_{\bar{z}}^{2}\left\langle z,\alpha z\right\rangle =0$ and~$\beta_{t}=u_{\alpha}\left(t,0\right)^{\vee2}\hat{\beta}_{t}$.
\end{proof}

We can thus restrict our proof to the case of a polynomial~$Q_{t}$
of the form~$Q_{t}\left(z\right)=\Im\left\langle \beta_{t},z^{\vee2}\right\rangle $
with~$\beta_{t}\in\mathcal{C}^{0}\left(\mathbb{R};\mathcal{Z}^{\vee2}\right)$
and no~$\left(\alpha_{t}\right)$ term.

\section{A Dyson type expansion formula for the Wick symbol of the evolved
quantum observable\label{sec:integral-formula}}

In this section we prove Theorem~\ref{thm:integral-formula}.
\begin{proof}
We first prove that the formula, for~$c\in\mathcal{P}_{\leq m}\left(\mathcal{Z}\right)$,\[
U\left(0,s\right)\left(c\circ\varphi\left(0,s\right)\right)^{Wick}U\left(s,0\right)=c^{Wick}-\frac{i\varepsilon}{2}\int_{0}^{s}U\left(0,\sigma\right)\left\{ c\circ\varphi\left(0,\sigma\right),Q_{\sigma}\right\} ^{\left(2\right)Wick}U\left(\sigma,0\right)d\sigma\]
holds as an equality of continuous operators from~$\mathcal{D}(\left\langle N\right\rangle ^{m/2})$
to~$\mathcal{H}$, with~$\left\langle N\right\rangle =(N^{2}+1)^{1/2}$.
This is a consequence of the fact that the derivative of the left
hand term as a function of~$s$ is~$-\frac{i\varepsilon}{2}U\left(0,s\right)\left\{ c\circ\varphi\left(0,s\right),Q_{s}\right\} ^{\left(2\right)Wick}U\left(s,0\right)$
as it can be seen from the relation\[
i\partial_{\sigma}\left(c\circ\varphi\left(0,\sigma\right)\right)=-\partial_{z}\left(c\circ\varphi\left(0,\sigma\right)\right).\partial_{\bar{z}}Q_{\sigma}+\partial_{z}Q_{\sigma}.\partial_{\bar{z}}\left(c\circ\varphi\left(0,\sigma\right)\right)\]
and Proposition~\ref{pro:Wick-calculus}. Applying the previous
formula with~$c=b^{\left(K\right)t,\bar{s}^{K}}$ we get recursively
\begin{eqnarray*}
 &  & U\left(0,t\right)b^{Wick}U\left(t,0\right)\\
 &  & \quad=\sum_{k=0}^{K-1}\left(\frac{\varepsilon}{2}\right)^{k}\int_{\bar{s}^{k}\in\Delta_{t}^{k}}\left(b^{\left(k\right)t,\bar{s}^{k}}\right)^{Wick}d\bar{s}^{k}\\
 &  & \quad\phantom{=}+\left(\frac{\varepsilon}{2}\right)^{K}\int_{\bar{s}^{K}\in\Delta_{t}^{K}}U\left(0,s_{K}\right)\left(b^{\left(K\right)t,\bar{s}^{K}}\circ\varphi\left(0,s_{K}\right)\right)^{Wick}U\left(s_{K},0\right)d\bar{s}^{K}\,.\end{eqnarray*}
This process gives a null remainder as soon as~$K>m/2$ as for~$K\leq\left\lfloor m/2\right\rfloor $,
since the polynomial~$b^{\left(K\right)\bar{s}^{K}}$ is of total
order~$m-2K$.
\end{proof}

\section{An exponential type expansion formula for the Wick symbol of the
evolved observable\label{sec:exponential-formula}}

In this section we prove Theorem~\ref{thm:exponential-formula}.

\subsection{Quantum evolution as a Bogoliubov implementation}

Some basic facts about symplectomorphisms are recalled in Appendix~\ref{sec:symplectic-transformations}.
\begin{defn}
\label{def:implementable-symplectomorphism}A symplectomorphism~$T$
is called \emph{implementable} if and only if there exists a unitary
operator~$U$ on~$\mathcal{H}$ , called a \emph{Bogoliubov implementer}
of~$T$, such that\[
\forall\xi\in\mathcal{Z},\, U^{*}W\left(\xi\right)U=W\left(T\xi\right)\,.\]
\end{defn}
\begin{prop}
\label{pro:bogoliubov-implementation}Assume~$\alpha_{t}\equiv0$
and~\textbf{H2}. Let~$Q_{t}=\Im\left\langle \beta_{t},z^{\vee2}\right\rangle $,
$\varphi\left(t,s\right)$ the associated classical evolution (see
Section~\ref{sec:Classical-evolution}) and~$U\left(t,s\right)$
the associated quantum evolution (see Section~\ref{sec:Existence-quantum-evolution}).
Then for all~$t$ in~$\mathbb{R}$, $U\left(t,0\right)$ is a Bogoliubov
implementer of~$-i\varphi\left(0,t\right)i$.\end{prop}
\begin{proof}
 We begin with a formal computation which will be justified further.
It suffices to show that \[
i\varepsilon\partial_{t}\left[U\left(0,t\right)W\left(-i\varphi\left(t,0\right)i\xi\right)U\left(t,0\right)\right]=0\,.\]
Computing this derivative and omitting the time and~$-i\varphi\left(t,0\right)i\xi$
dependencies in our notations, we get with~$U\left(t,0\right)=U$\[
U^{*}W\left\{ -W^{*}Q^{Wick}W+Q^{Wick}+W^{*}i\varepsilon\partial_{t}W\right\} U\,.\]
Then from Proposition~2.10~(iii) in~\cite{MR2465733}, the differential
formula of Weyl operators recalled in Proposition~\ref{pro:derivative-Weyl}
below  and with~$f_{t}=-i\varphi\left(t,0\right)i\xi$ it suffices
to show that \[
Q\left(z+\frac{i\varepsilon}{\sqrt{2}}f_{t}\right)=Q\left(z\right)+i\varepsilon\left(\frac{i\varepsilon}{2}\Im\left\langle f_{t},\partial_{t}f_{t}\right\rangle +i\sqrt{2}\Re\left\langle \partial_{t}f_{t},z\right\rangle \right)\]
to get the result. This equality results from the expansion of~$Q\left(z\right)=\Im\left\langle \beta,z^{\vee2}\right\rangle $,
recalling that $i\partial_{t}\varphi\left(t,0\right)\xi=\partial_{\bar{z}}Q\left(\varphi\left(t,0\right)\xi\right)$,
and observing that~$\partial_{\bar{z}}Q\left(z\right)=i\left(\left\langle z\right|\vee I_{\mathcal{Z}}\right)\left|\beta\right\rangle $.
We now need to clarify the meaning of this computation. It suffices
to show that the quantity\[
\left\langle \Phi,U\left(0,t\right)W\left(-i\varphi\left(t,0\right)i\xi\right)U\left(t,0\right)\Psi\right\rangle \]
is constant for~$\Psi$, $\Phi$ in $\mathcal{D}\left(N+1\right)$.
Since this domain is preserved by the operators~$U(t,s)$, the Weyl
operators are weakly derivable on this domain (see next proposition),
and~$U\left(t,s\right)$ is derivable on this domain, then we get
the justification of the previous formal computation.
\end{proof}

\begin{prop}
\label{pro:derivative-Weyl}Let~$z$,~$h$ be vectors in~$\mathcal{Z}$,~$t$
be a real parameter and~$\varphi$,~$\psi$ be in the domain of~$\Phi\left(h\right)$.
Then\begin{eqnarray*}
\lim_{t\to0}\frac{1}{t}\left(\left\langle \varphi,\left[W\left(z+th\right)-W\left(z\right)\right]\psi\right\rangle \right) & = & \left\langle \varphi,W\left(z\right)\left[i\Phi\left(h\right)+\frac{i\varepsilon}{2}\Im\left\langle z,h\right\rangle +\right]\psi\right\rangle \\
 & = & \left\langle \varphi,\left[i\Phi\left(h\right)-\frac{i\varepsilon}{2}\Im\left\langle z,h\right\rangle \right]W\left(z\right)\psi\right\rangle \,.\end{eqnarray*}
\end{prop}
\begin{proof}
For the first equality. The Weyl commutation relations give\begin{eqnarray*}
\frac{1}{t}\left\langle \varphi,\left[W(z+th)-W(z)\right]\psi\right\rangle  & = & \frac{1}{t}\left\langle W(-z)\varphi,\left[e^{\frac{i\varepsilon}{2}\Im\langle z,th\rangle}W(th)-I_{\mathcal{Z}}\right]\psi\right\rangle \\
 & = & \left\langle W(-z)\varphi,e^{\frac{i\varepsilon}{2}\Im\langle z,th\rangle}\frac{1}{t}(W(th)-I_{\mathcal{Z}})\psi\right\rangle \\
 &  & +\frac{1}{t}\left(e^{\frac{i\varepsilon}{2}\Im\left\langle z,th\right\rangle }-1\right)\left\langle W(-z)\varphi,\psi\right\rangle \\
 & \underset{t\to0}{\rightarrow} & \left\langle \varphi,W(z)\left[i\Phi\left(h\right)+\frac{i\varepsilon}{2}\Im\left\langle z,h\right\rangle \right]\psi\right\rangle \,.\end{eqnarray*}
The convergence of the first term is due to the continuous one parameter
group structure of~$W\left(th\right)$. The other equality is obtained
in the same way.
\end{proof}

\subsection{Action of Bogoliubov transformations on Wick symbols}

A theorem due to Shale~(see~\cite{MR0137504}) characterizes implementable
symplectomorphisms. We quote here a version of this theorem fitting
our needs.
\begin{thm}[Shale, 1962]
\label{thm:Shale}A symplectomorphism~$T$ is implementable if and
only if the $\mathbb{C}$-linear part of~$T^{*}T-Id$ is trace class.
\end{thm}
We can now quote the main result of this part.
\begin{thm}
\label{thm:fomula2-for-symplecto}Let~$T=L+A$ with~$L$ $\mathbb{C}$-linear
and~$A$ $\mathbb{C}$-antilinear, be an implementable symplectomorphism
with a Bogoliubov implementer~$U$ preserving~$\mathcal{D}(\left\langle N\right\rangle ^{k/2})$
for any integer~$k$$\geq2$, then for any polynomial~$b$ in~$\mathcal{P}_{\leq m}\left(\mathcal{Z}\right)$
with~$m\geq2$, \begin{equation}
U^{*}b^{Wick}U=\left(e^{\frac{\varepsilon}{2}\Lambda\left[T\right]}\left[b\left(T^{*}\cdot\right)\right]\right)^{Wick}\label{eq:formula2_for_symplecto}\end{equation}
as an equality of continuous operators from~$\mathcal{D}(\left\langle N\right\rangle ^{m/2})$
to~$\mathcal{H}$, with~$\left\langle N\right\rangle =(N^{2}+1)^{1/2}$,
where
\begin{itemize}
\item the exponential is a finite expansion whose rank depends on the degree
of the polynomial~$b$,
\item the operator~$\Lambda\left[T\right]$ is defined on any polynomial~c
by \[
\Lambda\left[T\right]c\left(z\right)=\mbox{Tr}\left[-2AA^{*}\partial_{\bar{z}}\partial_{z}c\left(z\right)\right]+\left\langle v\right|.\partial_{\bar{z}}^{2}c\left(z\right)+\partial_{z}^{2}c\left(z\right).\left|v\right\rangle \]
with~$v\in\bigotimes^{2}\mathcal{Z}$ the vector such that for all~$z_{1},\, z_{2}\in\mathcal{Z}$,
$\left\langle z_{1}\otimes z_{2},v\right\rangle =\left\langle z_{1},LA^{*}z_{2}\right\rangle $.
\end{itemize}
\end{thm}
In order to prove this result, we use intermediate steps.
\begin{enumerate}
\item We prove that~$U^{*}b^{Weyl}U=b\left(T^{*}\cdot\right)^{Weyl}$ in
finite dimension.
\item We use the Fourier transform and the formula\[
b^{Weyl}=\frac{1}{\left(\pi\varepsilon/2\right)^{d}}\left(b*e^{-\frac{\left|z\right|^{2}}{\varepsilon/2}}\right)^{Wick}\]
 to get the result in finite dimension.
\item We extend the result to infinite dimension.
\end{enumerate}

\subsubsection{Action of Bogoliubov transformations on Weyl quantizations of polynomials
in finite dimension}

\begin{defn}
In a finite-dimensional Hilbert space~$\mathcal{Z}$ identified with~$\mathbb{C}^{d}$,
the \emph{symplectic Fourier transform} is defined by \[
\mathcal{F}^{\sigma}\left[f\right]\left(z\right)=\int_{\mathcal{Z}}e^{2\pi i\sigma\left(z,z'\right)}f\left(z'\right)L\left(dz'\right)\]
where~$L$ denotes the Lebesgue measure, and~$f$ is any Schwartz
tempered distribution. We associate with each polynomial~$b\in\mathcal{P}_{p,q}\left(\mathcal{Z}\right)$
a Weyl observable by \begin{equation}
b^{Weyl}=\int_{\mathcal{Z}}\mathcal{F}^{\sigma}\left[b\right]\left(z\right)W\left(-i\sqrt{2}\pi z\right)L\left(dz\right)\,.\label{eq:def-quantif-Weyl}\end{equation}
This formula has a meaning as an equality of quadratic forms on~$\mathcal{S}\left(\mathcal{Z}\right)$
since for any~$\Phi$, $\Psi$ in~$\mathcal{S}\left(\mathcal{Z}\right)$,
$z\mapsto\left\langle \Phi,W(-i\sqrt{2}\pi z)\Psi\right\rangle $
and its derivative are continuous bounded functions and~$\mathcal{F}^{\sigma}\left[b\right]$
is made of derivatives of the delta function.
\end{defn}

\begin{prop}
\label{pro:formula-Weyl-finite-dim}Let~$b\in\mathcal{P}_{\leq m}\left(\mathcal{Z}\right)$
with~$m\geq2$ be a polynomial on a finite-dimensional Hilbert space~$\mathcal{Z}$.
Let~$T$ be an implementable symplectomorphism with implementation~$U$
preserving the domain~$\mathcal{D}(\left\langle N\right\rangle ^{m/2})$.
Then \[
U^{*}b^{Weyl}U=b\left(T^{*}\cdot\right)^{Weyl}\]
as a continuous operator from~$\mathcal{D}(\left\langle N\right\rangle ^{m/2})$
to~$\mathcal{H}$.\end{prop}
\begin{proof}
We compute, in the sense of quadratic forms on~$\mathcal{S}\left(\mathcal{Z}\right)$,\begin{eqnarray*}
U^{*}b^{Weyl}U & = & \int\mathcal{F}^{\sigma}\left[b\right]\left(z\right)W\left(-\sqrt{2}\pi Tiz\right)L\left(dz\right)\\
 & = & \int\mathcal{F}^{\sigma}\left[b\right]\left(T^{*}z\right)W\left(-i\sqrt{2}\pi z\right)L\left(dz\right)\\
 & = & \int\mathcal{F}^{\sigma}\left[b\left(T^{*}\cdot\right)\right]\left(z\right)W\left(-i\sqrt{2}\pi z\right)L\left(dz\right)\\
 & = & b\left(T^{*}\cdot\right)^{Weyl}\end{eqnarray*}
where we made use of the relation~$Ti=i\left(T^{*}\right)^{-1}$,
the volume preservation of~$T^{*}$ in~$\mathcal{Z}$ seen as a
$\mathbb{R}$-vector space and the property of composition of a symplectic
Fourier transform by a symplectomorphism (see Appendix~\ref{sec:Symplectic-Fourier-transform}).
The boundedness from~$\mathcal{D}(\left\langle N\right\rangle ^{m/2})$
to~$\mathcal{H}$ is deduced from the facts that the Fourier transform
of~$b$ involves only derivatives of the delta function of order
smaller or equal to~$m$ and that a derivation of the Weyl operator
gives at worse a field factor which is controlled by~$\left\langle N\right\rangle ^{1/2}$.
\end{proof}

\subsubsection{Action of Bogoliubov transformations on Wick quantization of polynomials
in finite dimension}

\begin{prop}
\label{pro:formula-Wick-finite-dim}Let~$b\in\mathcal{P}_{\leq m}\left(\mathcal{Z}\right)$
with~$m\geq2$ be a polynomial on a finite-dimensional Hilbert space~$\mathcal{Z}$.
Let~$T$ be an implementable symplectomorphism with implementation~$U$preserving
the domain~$\mathcal{D}(\left\langle N\right\rangle ^{m/2})$. Then
\begin{equation}
U^{*}b^{Wick}U=\left(e^{\frac{\varepsilon}{2}\Lambda\left[T\right]}\left[b\left(T^{*}\cdot\right)\right]\right)^{Wick}\,,\label{eq:finite_dim_formula}\end{equation}
as a continuous operator from~$\mathcal{D}(\left\langle N\right\rangle ^{m/2})$
to~$\mathcal{H}$, where~$\Lambda\left[T\right]$ is defined as
in Theorem~\ref{thm:fomula2-for-symplecto}.\end{prop}
\begin{proof}
We search the polynomial~$c$ such that~$U^{*}b^{Wick}U=c^{Wick}$.
In finite dimension for polynomials we can use the well known deconvolution
formula \[
c^{Wick}=\left(c*\frac{1}{\left(\pi\varepsilon/2\right)^{d}}e^{\frac{\left|z\right|^{2}}{\varepsilon/2}}\right)^{Weyl}\,.\]
By Proposition~\ref{pro:formula-Weyl-finite-dim} we boil down to
search for a polynomial~$c$ such that \[
\left(b*\frac{1}{\left(\pi\varepsilon/2\right)^{d}}e^{\frac{\left|z\right|^{2}}{\varepsilon/2}}\right)\left(T^{*}\cdot\right)=c*\frac{1}{\left(\pi\varepsilon/2\right)^{d}}e^{\frac{\left|z\right|^{2}}{\varepsilon/2}}\,.\]
Using symplectic Fourier transform (see appendix~\ref{sec:Symplectic-Fourier-transform})
and its properties with respect to convolution, composition with symplectomorphisms
and Gaussians, we get\begin{eqnarray*}
\mathcal{F}^{\sigma}c & = & \left[\mathcal{F}^{\sigma}b\left(T^{*}\cdot\right)\right]\times\left[\mathcal{F}^{\sigma}\left(\frac{e^{\frac{\left|z\right|^{2}}{\varepsilon/2}}}{\left(\pi\varepsilon/2\right)^{d}}\right)\left(T^{*}\cdot\right)\right]\times\left[\mathcal{F}^{\sigma}\left(\frac{e^{-\frac{\left|z\right|^{2}}{\varepsilon/2}}}{\left(\pi\varepsilon/2\right)^{d}}\right)\right]\\
 & = & e^{\frac{\pi^{2}\varepsilon\left(\left|T^{*}\cdot\right|^{2}-\left|\cdot\right|^{2}\right)}{2}}\times\mathcal{F}^{\sigma}b\left(T^{*}\cdot\right)\,.\end{eqnarray*}
 Writting~$T=L+A$ with~$L$ the $\mathbb{C}$-linear and~$A$
the $\mathbb{C}$-antilinear part of~$T$ we obtain \begin{eqnarray*}
\left|T^{*}z\right|^{2}-\left|z\right|^{2} & = & \left\langle L^{*}z,L^{*}z\right\rangle +\left\langle A^{*}z,A^{*}z\right\rangle +\left\langle L^{*}z,A^{*}z\right\rangle +\left\langle A^{*}z,L^{*}z\right\rangle -\left\langle z,z\right\rangle \\
 & = & \left\langle z,LL^{*}z\right\rangle +\left\langle z,AA^{*}z\right\rangle +\left\langle LA^{*}z,z\right\rangle +\left\langle z,LA^{*}z\right\rangle -\left\langle z,z\right\rangle \\
 & = & \left\langle z,2AA^{*}z\right\rangle +\left\langle v,z^{\vee^{2}}\right\rangle +\left\langle z^{2},v\right\rangle \end{eqnarray*}
with~$v\in\bigotimes^{2}\mathcal{Z}$ the vector such that for all~$z_{1},\, z_{2}\in\mathcal{Z}$,
$\left\langle z_{1}\otimes z_{2},v\right\rangle =\left\langle z_{1},LA^{*}z_{2}\right\rangle $.
  By Fourier transforming again, we get \[
\pi^{2}\mathcal{F}^{\sigma}\left[\left(\left|T^{*}\cdot\right|^{2}-\left|\cdot\right|^{2}\right)\times\cdot\right]\mathcal{F}^{\sigma}c=\mbox{Tr}\left[-2AA^{*}\partial_{\bar{z}}\partial_{z}c\left(z\right)\right]+\left\langle v\right|\partial_{\bar{z}}^{2}c\left(z\right)+\partial_{z}^{2}c\left(z\right)\left|v\right\rangle \]
as the $\mathbb{C}$-linear and $\mathbb{C}$-antilinear parts behave
differently under Fourier transform (the $\mathbb{C}$-linear part
has a minus sign added, see appendix~\ref{sec:Symplectic-Fourier-transform}).
We then obtain the claimed result.
\end{proof}

\subsubsection{Extension to infinite dimension on a ``cylindrical'' class of polynomials}
\begin{thm}
Let~$\hat{T}$ be symplectomorphism of the form~$\hat{T}=e^{c\rho}$,
with~$c$ a conjugation and~$\rho$ a positive, self-adjoint, Hilbert-Schmidt
operator commuting with~$c$. Let~$\left(\xi_{j}\right)_{j\in\mathbb{N}}$
a Hilbert basis in which~$\rho$ is diagonal. Let~$\pi_{K}$ be
the orthogonal projection on the finite-dimensional space~$\mathcal{Z}_{K}=Vect(\left\{ \xi_{j}\right\} _{j\leq K})$.

Then for any polynomial~$b$ in~$\mathcal{P}_{m}\left(\mathcal{Z}\right)$
with~$m\geq2$ and any integer~$K$\[
\hat{U}^{*}b_{K}^{Wick}\hat{U}=\left(e^{\frac{\varepsilon}{2}\Lambda\left[\hat{T}\right]}\left[b_{K}\left(\hat{T}^{*}\cdot\right)\right]\right)^{Wick}\]
as continuous operators from~$\mathcal{D}(\left\langle N\right\rangle ^{\frac{m}{2}})$
to~$\mathcal{H}$ where~$b_{K}\left(z\right)=b\left(\pi_{K}z\right)$.\end{thm}
\begin{proof}
We first remark that, with~$Q\left(z\right)=\Im\left\langle c\rho z,z\right\rangle $,
$e^{-iQ^{Wick}/\varepsilon}$ is a Bogoliubov implementer of~$\hat{T}$
as it can be seen using Proposition~\ref{pro:bogoliubov-implementation}
and the Hilbert-Schmidt property of~$\rho$. We define~$\rho_{L}=\rho\pi_{L}$,
$\hat{T}_{L}=\hat{T}\pi_{L}$ and the operator~$Q_{L}\left(z\right)^{Wick}=\Im\left\langle c\rho_{L}z,z\right\rangle ^{Wick}$.
We use the identification~$\mathcal{H}=\Gamma_{s}\left(\mathcal{Z}_{L}\right)\otimes\Gamma_{s}(\mathcal{Z}_{L}^{\bot})$
and observe that on~$\Gamma_{s}\left(\mathcal{Z}_{L}\right)\otimes\{\Omega^{\mathcal{Z}_{L}^{\bot}}\}$,
$e^{-iQ^{Wick}/\varepsilon}=e^{-iQ_{L}^{Wick}/\varepsilon}$. For~$K\leq L$
we obtain on~$\Gamma_{s}\left(\mathcal{Z}_{L}\right)\otimes\{\Omega^{\mathcal{Z}_{L}^{\bot}}\}$\[
\hat{U}_{L}^{*}b_{K}^{Wick}\hat{U}_{L}=\left(e^{\frac{\varepsilon}{2}\Lambda\left[\hat{T}_{L}\right]}\left[b_{K}\left(\hat{T}_{L}^{*}\cdot\right)\right]\right)^{Wick}\]
by Proposition~\ref{pro:formula-Wick-finite-dim}, with $\hat{U}_{L}=e^{-iQ_{L}^{Wick}/\varepsilon}$.
But on this domain it is the same as

\[
\hat{U}^{*}b_{K}^{Wick}\hat{U}=\left(e^{\frac{\varepsilon}{2}\Lambda\left[\hat{T}\right]}\left[b_{K}\left(\hat{T}^{*}\cdot\right)\right]\right)^{Wick}\]
with~$\hat{U}=e^{-iQ^{Wick}/\varepsilon}$. We thus get an equality
on~$\cup_{L}\Gamma_{s}\left(\mathcal{Z}_{L}\right)$, and by continuity
of the involved operators from~$\mathcal{D}(\left\langle N\right\rangle ^{\frac{m}{2}})$
to~$\mathcal{H}$ we get the expected result.
\end{proof}

We will first show that Formula~(\ref{eq:formula2_for_symplecto})
apply in particular to a well chosen class of cylindrical polynomials,
and then extend it by density to every polynomial.

\subsubsection{Extension to general polynomials}

We split the proof of Formula~(\ref{eq:formula2_for_symplecto})
for general polynomials into several lemmata and propositions.
\begin{lem}
\label{lem:CVfaible_opdimfinie}Let~$\left(\xi_{j}\right)_{j\in\mathbb{N}}$
be a Hilbert basis of~$\mathcal{Z}$,~$\pi_{m}$ be the orthogonal
projector on~$\mathcal{Z}_{m}=Vect(\left\{ \xi_{j}\right\} _{j\leq m})$.
Let~$b$ be a polynomial in~$\mathcal{P}_{p,q}\left(\mathcal{Z}\right)$
and define~$b_{K}=b\left(\pi_{K}\cdot\right)$. Then~$(\widetilde{b_{K}})_{K\in\mathbb{N}}$
is bounded and \[
\tilde{b}=w-\lim_{j\to\infty}\widetilde{b_{K}}\,.\]

\end{lem}

To formulate more clearly some convergence results we need some extra
definitions.
\begin{defn}
We define the spaces\[
\mathcal{L}_{p,q}^{\vee}\left(\mathcal{Z}\right)=\mathcal{L}\left(\mathcal{Z}^{\vee p},\mathcal{Z}^{\vee q}\right)\,,\quad\mathcal{L}_{m}^{\vee}=\bigoplus_{p+q=m}\mathcal{L}_{p,q}^{\vee}\quad\mbox{and}\quad\mathcal{L}_{\leq m}^{\vee}=\bigoplus_{m'\leq m}\mathcal{L}_{m'}^{\vee}\]
corresponding to~$\mathcal{P}_{p,q}\left(\mathcal{Z}\right)$, $\mathcal{P}_{m}\left(\mathcal{Z}\right)$
and~$\mathcal{P}_{\leq m}\left(\mathcal{Z}\right)$.

Let~$b=\sum_{p,q}b_{p,q}$ be a polynomial, with~$b_{p,q}\in\mathcal{P}\left(\mathcal{Z}\right)$.
We note~$\tilde{b}=(\widetilde{b_{p,q}})\in\bigoplus_{p,q}\mathcal{L}_{p,q}^{\vee}\left(\mathcal{Z}\right)$.

The norm of~$\tilde{b}=(\widetilde{b_{p,q}})\in\mathcal{L}_{\leq m}^{\vee}\left(\mathcal{Z}\right)$
is~$\|\tilde{b}\|_{\mathcal{L}_{\leq m}^{\vee}\left(\mathcal{Z}\right)}=\sum_{p,q}\|\widetilde{b_{p,q}}\|_{\mathcal{L}\left(\bigvee^{p}\mathcal{Z},\bigvee^{q}\mathcal{Z}\right)}\,.$

A sequence~$(\tilde{b}_{K})_{K\in\mathbb{N}}$ of elements of~$\mathcal{L}_{\leq m}\left(\mathcal{Z}\right)$
\emph{converges weakly} to~$\tilde{b}$ in~$\mathcal{L}_{\leq m}\left(\mathcal{Z}\right)$
if~$\tilde{b}_{K_{p,q}}$ converges weakly to~$\tilde{b}_{p,q}$
for every~$p$ and~$q$ as~$K\to+\infty$. \end{defn}
\begin{lem}
Let~$T$ be an operator in~$\mathcal{X}\left(\mathcal{Z}\right)$,
$\left(b_{K}\right)_{K\in\mathbb{N}}$ and~$b$ be polynomials in~$\mathcal{P}_{m}\left(\mathcal{Z}\right)$
such that~$(\tilde{b}_{K})_{K\in\mathbb{N}}$ converges weakly to~$\tilde{b}$.
Then~$ $$b_{K}\left(T\cdot\right)$ and~$b\left(T\cdot\right)$
are in~$\mathcal{P}_{m}\left(\mathcal{Z}\right)$ and~$\widetilde{b_{K}\left(T\cdot\right)}$
converges weakly to~$\widetilde{b\left(T\cdot\right)}$.
\end{lem}

\begin{lem}
Let~$T$ be an operator in~$\mathcal{X}\left(\mathcal{Z}\right)$,
$\left(b_{K}\right)_{K\in\mathbb{N}}$ and~$b$ be polynomials in~$\mathcal{P}_{m}\left(\mathcal{Z}\right)$
such that~$(\tilde{b}_{K})_{K\in\mathbb{N}}$ is bounded and converges
weakly to~$\tilde{b}$. Then~$(\widetilde{e^{\frac{\varepsilon}{2}\Lambda\left[T\right]}b_{K}})_{K\in\mathbb{N}}$
converges weakly to~$\widetilde{e^{\frac{\varepsilon}{2}\Lambda\left[T\right]}b}$.\end{lem}
\begin{proof}
It is enough to show that weak convergence is preserved by the action
of~$\Lambda\left[T\right]$. But, for any polynomial~$b$, \[
\widetilde{\Lambda\left[T\right]b}=\Tr_{1}\left[\left(-2A^{*}A\otimes I_{\mathcal{Z}^{\vee q-1}}\right)\tilde{b}\right]+\left(\left\langle v\right|\vee I_{\mathcal{Z}^{\vee q-2}}\right)\tilde{b}+\tilde{b}\left(\left|v\right\rangle \vee I_{\mathcal{Z}^{\vee p-2}}\right)\,,\]
where~$\Tr_{1}$ is the partial trace on the first~$\mathcal{Z}$
subspace on the left and any direction on the right (so that if~$\tilde{b}\in\mathcal{L}_{p,q}^{\vee}\left(\mathcal{Z}\right)$,
then~$\Tr_{1}[(-2A^{*}A\otimes I_{\mathcal{Z}^{\vee q-1}})\tilde{b}]$
is in~$\mathcal{L}_{p-1,q-1}^{\vee}\left(\mathcal{Z}\right)$). With
this formula the preservation of the weak convergence is clear.
\end{proof}

\begin{prop}
Let~$b$ and~$\left(b_{K}\right)_{K\in\mathbb{N}}$ be Wick polynomials
in~$\mathcal{P}_{p,q}\left(\mathcal{Z}\right)$ such that~$w-\lim\tilde{b}_{K}=\tilde{b}$.
Then \[
w-\lim_{K}\left(b_{K}-b\right)^{Wick}\left\langle N\right\rangle ^{-\frac{p+q}{2}}=0\,.\]

\end{prop}

\begin{prop}
\label{pro:cv-UbU}Let~$b$ and~$\left(b_{K}\right)_{K\in\mathbb{N}}$
be Wick polynomials in~$\mathcal{P}_{p,q}\left(\mathcal{Z}\right)$
such that~$w-\lim\tilde{b}_{K}=\tilde{b}$. Let~$U$ be a unitary
operator on the Fock space~$\mathcal{H}$ such that, for all~$k\geq2$,
$\left\langle N\right\rangle ^{\frac{k}{2}}U\left\langle N\right\rangle ^{-\frac{k}{2}}$
is a bounded operator. Then\[
w-\lim_{K}U^{*}\left(b_{K}-b\right)^{Wick}U\left\langle N\right\rangle ^{-\frac{m'}{2}}=0\]
with~$m'=\max\left(m,2\right)$, $m=p+q$.
\end{prop}

\begin{prop}
\label{pro:formula-Wick-infinite-dim}Let~$T$ be an implementable
symplectomorphism with Bogoliubov implementer~$U$. Then for any
polynomial~$b$ in~$\mathcal{P}_{\leq m}\left(\mathcal{Z}\right)$,
$m\geq2$, \[
U^{*}b^{Wick}U=\left(e^{\frac{\varepsilon}{2}\Lambda\left[T\right]}\left[b\left(T^{*}\cdot\right)\right]\right)^{Wick}\]
as continuous operators from~$\mathcal{D}(\left\langle N\right\rangle ^{\frac{m}{2}})$
to~$\mathcal{H}$.\end{prop}
\begin{proof}
From the results~\ref{lem:CVfaible_opdimfinie} to~\ref{pro:cv-UbU}
we deduce the result for symplectomorphisms of the form~$\hat{T}=e^{c\rho}$,
with~$c$ a conjugation and~$\rho$ a positive, self-adjoint, Hilbert-Schmidt
operator commuting with~$c$. Then we observe that the hypothesis
on the form of~$\hat{T}$ is not restrictive as, if $T$ is of the
form~$ue^{c\rho}$ with~$u$ unitary, then, with~$\hat{U}$ a Bogoliubov
implementer for~$\hat{T}$, $\hat{U}\Gamma\left(u^{*}\right)$ is
a Bogoliubov implementer for~$T$ and \[
\Gamma\left(u\right)\left(e^{\frac{\varepsilon}{2}\Lambda\left[\hat{T}\right]}\left[b\left(\hat{T}^{*}\cdot\right)\right]\right)^{Wick}\Gamma\left(u^{*}\right)=\left(e^{\frac{\varepsilon}{2}\Lambda\left[T\right]}\left[b\left(T^{*}\cdot\right)\right]\right)^{Wick}\,.\]
Indeed for any polynomial~$c$, and operator~$\varphi$ in~$\mathcal{X}\left(\mathcal{Z}\right)$,
$\Gamma\left(\varphi\right)c^{Wick}\Gamma\left(\varphi^{*}\right)=c\left(\varphi^{*}\cdot\right)^{Wick}$
and\[
\Lambda\left[\hat{T}\right]^{k}\left[b\left(\hat{T}^{*}\cdot\right)\right]\left(u^{*}\cdot\right)=\Lambda\left[u\hat{T}\right]^{k}\left[b\left(\hat{T}^{*}u^{*}\cdot\right)\right]\]
 as can be checked by an explicit computation and using the fact that~$L=u\hat{L}$
and~$A=u\hat{A}$ with the~$L$, $\hat{L}$ and~$A$, $\hat{A}$
denoting respectively the $\mathbb{C}$-linear and $\mathbb{C}$-antilinear
parts of~$T$ and~$\hat{T}$. This achieves the proof.
\end{proof}

\subsection{An evolution formula for the Wick symbol}

We can now prove Theorem~\ref{thm:exponential-formula}.
\begin{proof}
We only need to apply propositions~\ref{pro:bogoliubov-implementation}
and~\ref{pro:formula-Wick-infinite-dim} with~$T=-i\varphi\left(0,t\right)i=L^{*}\left(t,0\right)+A^{*}\left(t,0\right)$
(with~$\varphi\left(t,0\right)=L\left(t,0\right)+A\left(t,0\right)$).
We remark that for any symplectomorphism~$T$, $\left(-iTi\right)^{*}=T^{-1}$
so that $\left(-i\varphi\left(0,t\right)i\right)^{*}=\varphi\left(t,0\right)$
and thus we get the result.
\end{proof}

\subsection{Estimates}

We now give estimates for the different terms of the expansion of
the symbol.
\begin{prop}
\label{pro:estimate_Lambda_T}Let~$T=L+A$ be an implementable symplectomorphism
with~$L$ $\mathbb{C}$-linear and~$A$ $\mathbb{C}$-antilinear.
Then the operator~$\Lambda\left[T\right]$ defined on~$\mathcal{P}\left(\mathcal{Z}\right)$
by \[
\Lambda\left[T\right]c\left(z\right)=Tr\left[-2AA^{*}\partial_{\bar{z}}\partial_{z}c\right]+\left\langle v\right|\partial_{\bar{z}}^{2}c\left(z\right)+\partial_{z}^{2}c\left(z\right)\left|v\right\rangle \,,\]
with~$v\in\bigotimes^{2}\mathcal{Z}$ the vector such that for all~$z_{1},\, z_{2}\in\mathcal{Z}$,
$\left\langle z_{1}\otimes z_{2},v\right\rangle =\left\langle z_{1},LA^{*}z_{2}\right\rangle $
is such that, for~$c$ in~$\mathcal{P}_{m}\left(\mathcal{Z}\right)$\[
\left\Vert \Lambda\left[T\right]c\right\Vert _{\mathcal{P}_{m-2}\left(\mathcal{Z}\right)}\leq2\left\Vert T\right\Vert _{\mathcal{X}\left(\mathcal{Z}\right)}\left\Vert A\right\Vert _{\mathcal{L}_{2}^{a}\left(\mathcal{Z}\right)}\left\Vert c\right\Vert _{\mathcal{P}_{m}\left(\mathcal{Z}\right)}\,.\]
\end{prop}
\begin{proof}
We only have to remark that for any polynomial~$c$ in~$\mathcal{P}_{p,q}\left(\mathcal{Z}\right)$
the following estimates hold\[
\left\Vert Tr\left[B\partial_{\bar{z}}\partial_{z}c\left(z\right)\right]\right\Vert _{q-1\leftarrow p-1}\leq\left\Vert B\right\Vert _{\mathcal{L}_{1}\left(\mathcal{Z}\right)}\left\Vert c\right\Vert _{q\leftarrow p}\]
for any trace class operator~$B$, and \[
\left\Vert \left\langle v\right|\partial_{\bar{z}}^{2}c\left(z\right)\right\Vert _{q-2\leftarrow p}\leq\left\Vert v\right\Vert _{\bigvee^{2}\mathcal{Z}}\left\Vert c\right\Vert _{q\leftarrow p}\]
and that~$\left\Vert v\right\Vert _{\bigvee^{2}\mathcal{Z}}=\left\Vert LA^{*}\right\Vert _{\mathcal{L}_{2}^{a}\left(\mathcal{Z}\right)}\leq\left\Vert L\right\Vert _{\mathcal{L}\left(\mathcal{Z}\right)}\left\Vert A\right\Vert _{\mathcal{L}_{2}^{a}\left(\mathcal{Z}\right)}$.
The same estimate holds for~$\partial_{z}^{2}c\left(z\right)\left|v\right\rangle $.
\end{proof}
We apply this result to the expression given in the theorem~\ref{thm:exponential-formula}.
\begin{prop}
Let~$\left(Q_{t}\right)_{t}$ be a continuous one parameter family
of quadratic polynomials, $\varphi$ the classical flow associated
to~$\left(Q_{t}\right)_{t}$, and~$\Lambda^{t}$ the operator defined
in theorem~\ref{thm:exponential-formula}. Then, for~$b$ in~$\mathcal{P}_{\leq m}\left(\mathcal{Z}\right)$
\[
\left\Vert e^{\frac{\varepsilon}{2}\Lambda^{t}}\left(b\circ\varphi\left(t,0\right)\right)\right\Vert _{\mathcal{P}\left(\mathcal{Z}\right)}\leq\left\Vert b\right\Vert _{\mathcal{P}\left(\mathcal{Z}\right)}\left\Vert \varphi\left(t,0\right)\right\Vert _{\mathcal{X}\left(\mathcal{Z}\right)}^{m}\sum_{k=0}^{m}\frac{1}{k!}\left(\varepsilon\left\Vert \varphi\left(t,0\right)\right\Vert _{\mathcal{X}\left(\mathcal{Z}\right)}\left\Vert A\left(t,0\right)\right\Vert _{\mathcal{L}_{2}^{a}\left(\mathcal{Z}\right)}\right)^{k}\]
 where~$A$ is the $\mathbb{C}$-antilinear part of~$\varphi$.\end{prop}
\begin{proof}
It is enough to combine the propositions~\ref{pro:estimate_comp_classical_flow}
and~\ref{pro:estimate_Lambda_T}.\end{proof}
\begin{rem}
The norm~$\left\Vert \varphi\left(t,0\right)\right\Vert _{\mathcal{X}\left(\mathcal{Z}\right)}$
is bigger than~$1$ as for any symplectic transformation~$T=L+A$
with~$L$ $\mathbb{C}$-linear and~$A$ $\mathbb{C}$-antilinear,
$L^{*}L=I_{\mathcal{Z}}+A^{*}A\geq I_{\mathcal{Z}}$ (see proposition~\ref{pro:caracterisations_symplectique})
and thus~$\left\Vert T\right\Vert _{\mathcal{X}\left(\mathcal{Z}\right)}\geq\left\Vert L\right\Vert _{\mathcal{L}\left(\mathcal{Z}\right)}\geq1$.
\end{rem}
\appendix

\section{$\mathbb{R}$-linear symplectic transformations\label{sec:symplectic-transformations}}

In this part we adapt and recall some results of~\cite{MR628382}
to fit our needs.

Let~$\left(\mathcal{Z},\left\langle \cdot,\cdot\right\rangle \right)$
be a separable Hilbert space over the complex numbers field~$\mathbb{C}$.
The scalar products is linear with respect to the right variable and
antilinear with respect to the left variable. We note~$\mbox{Aut}_{\mathbb{R}}\left(\mathcal{Z}\right)$
the group of $\mathbb{R}$-linear continuous automorphisms on~$\mathcal{Z}$.
We define a symplectic form~$\sigma$ on~$\mathcal{Z}$ by\[
\sigma\left(z_{1},z_{2}\right):=\Im\left\langle z_{1},z_{2}\right\rangle \,.\]

\begin{defn}
A $\mathbb{R}$-linear automorphism~$T$ is a \emph{symplectomorphism}
if it preserves the symplectic form, i.e. if\[
\forall z_{1},z_{2}\in\mathcal{Z},\quad\sigma\left(Tz_{1},Tz_{2}\right)=\sigma\left(z_{1},z_{2}\right)\,.\]
We note~$\mbox{Sp}_{\mathbb{R}}\left(\mathcal{Z}\right)$ the set
of symplectic transformations over the Hilbert space~$\mathcal{Z}$.
It is a subgroup of~$\mbox{Aut}_{\mathbb{R}}\left(\mathcal{Z}\right)$.\end{defn}
\begin{prop}
\label{pro:decomp_symplecto1}A $\mathbb{R}$-linear application~$T:\mathcal{Z}\to\mathcal{Z}$
can be written as a sum of two applications respectively $\mathbb{C}$-linear
and $\mathbb{C}$-antilinear in a unique way : \[
T=\frac{T-iTi}{2}+\frac{T+iTi}{2}\,.\]

\end{prop}

\begin{defn}
Let~$A$ be a (bounded) $\mathbb{C}$-antilinear operator on the
Hilbert space~$\mathcal{Z}$. We define its \emph{adjoint}~$A^{*}$
as the only antilinear operator such that\[
\forall z_{1},z_{2}\in\mathcal{Z},\quad\left\langle z_{1},Az_{2}\right\rangle =\left\langle z_{2},A^{*}z_{1}\right\rangle \,.\]
 Let~$T=L+A=\mathcal{Z}\to\mathcal{Z}$ be a $\mathbb{R}$-linear
application with~$L$ $\mathbb{C}$-linear and~$A$ $\mathbb{C}$-antilinear.
The \emph{adjoint}~$T^{*}$ of~$T$ is defined by~$T^{*}=L^{*}+A^{*}$.\end{defn}
\begin{prop}
\label{pro:caracterisations_symplectique}Let~$T=L+A$ be a $\mathbb{R}$-linear
automorphism with~$L$ $\mathbb{C}$-linear and~$A$ $\mathbb{C}$-antilinear,
then the following conditions are equivalent.
\begin{enumerate}
\item $L+A$ is a symplectomorphism.
\item $\left(L^{*}-A^{*}\right)\left(L+A\right)=I_{\mathcal{Z}}$.
\item $\left(L^{*}+A^{*}\right)\left(L-A\right)=I_{\mathcal{Z}}$.
\item $L^{*}L-A^{*}A=I_{\mathcal{Z}}$ and $L^{*}A=A^{*}L$.\label{enu:characterzations-symplec-4}
\item $L^{*}-A^{*}$ is a symplectomorphism.
\item $L-A$ is a symplectomorphism. 
\item $LL^{*}-AA^{*}=I_{\mathcal{Z}}$ and $A^{*}L=L^{*}A$.
\end{enumerate}
\end{prop}
\begin{proof}
$\left(1\right)\Leftrightarrow\left(2\right)$ Let~$T=L+A$ a symplectomorphism,
for all~$z_{1},z_{2}\in\mathcal{Z}$,\begin{eqnarray*}
\sigma\left(z_{1},z_{2}\right)=\Im\left\langle z_{1},z_{2}\right\rangle  & = & \Im\left\langle \left(L+A\right)z_{1},Tz_{2}\right\rangle \\
 & = & \Im\left(\left\langle z_{1},L^{*}Tz_{2}\right\rangle +\overline{\left\langle z_{1},A^{*}Tz_{2}\right\rangle }\right)\\
 & = & \Im\left\langle z_{1},\left(L^{*}-A^{*}\right)Tz_{2}\right\rangle \,.\end{eqnarray*}
Replacing~$z_{1}$ by~$iz_{1}$ we get the same relation with a
real part instead of an imaginary part and finally\[
\left\langle z_{1},\left[\left(L^{*}-A^{*}\right)\left(L+A\right)-I_{\mathcal{Z}}\right]z_{2}\right\rangle =0\]
and this in turn implies~$\left(L^{*}-A^{*}\right)\left(L+A\right)=I_{\mathcal{Z}}$.
We can reverse the order of these calculations in order to obtain
the first equivalence.

$\left(2\right)\Leftrightarrow\left(3\right)$ The $\mathbb{C}$-linearity
and antilinearity properties of~$L$ and~$A$ give \[
\left(L^{*}-A^{*}\right)\left(L+A\right)i=i\left(L^{*}+A^{*}\right)\left(L-A\right)\]
so that we get the equivalent condition~$\left(3\right)$. 

$\left(\left(2\right)\,\mbox{and}\,\left(3\right)\right)\Leftrightarrow\left(4\right)$
The sum and the difference of the equations of~$\left(2\right)$
and~$\left(3\right)$ give~$\left(4\right)$ and the sum and difference
of the equations in~$\left(4\right)$ give~$\left(2\right)$ and~$\left(3\right)$.

$\left(1\right)\Leftrightarrow\left(5\right)$ From~$\left(1\right)\,\mbox{and}\,\left(3\right)$
we know that the inverse of a symplectomorphism~$T=L+A$ is $T^{-1}=L^{*}-A^{*}$
which is necessarily a symplectomorphism too, and thus~$ $$\left(1\right)\Rightarrow\left(5\right)$.
We get~$\left(5\right)\Rightarrow\left(1\right)$ exchanging~$T$
and~$T^{-1}$.

$\left(1\right)\Leftrightarrow\left(6\right)\Leftrightarrow\left(7\right)$
is easily deduced from the previous equivalences.\end{proof}
\begin{prop}
Let~$T=L+A$ be a symplectomorphism with~$L$ $\mathbb{C}$-linear
and~$A$ $\mathbb{C}$-antilinear, then~$L$ is invertible.\end{prop}
\begin{proof}
From Proposition~\ref{pro:caracterisations_symplectique} we get
\[
L^{*}L=I_{\mathcal{Z}}+A^{*}A\geq I_{\mathcal{Z}}\quad\mbox{and}\quad LL^{*}=I_{\mathcal{Z}}+AA^{*}\geq I_{\mathcal{Z}}\]
 and $ $thus~$L$ and~$L^{*}$ are injective. From the injectivity
of~$L^{*}$ we get~$\overline{\mbox{Ran}L}=\left(\mbox{Ker}L^{*}\right)^{\perp}=\left\{ 0\right\} ^{\perp}=\mathcal{Z}$. 

It is now enough to show that the range of~$L$ is closed. Pick a
vector~$y\in\mathcal{Z}$, there is a sequence~$\left(x_{n}\right)\in\mathcal{Z}^{\mathbb{N}}$
such that~$Lx_{n}\rightarrow y$. The relation~$L^{*}L\geq I_{\mathcal{Z}}$
gives~$\left|Lx_{m}-Lx_{n}\right|\geq\left|x_{n}-x_{m}\right|$.
The left hand part of the inequality goes to~$0$ for~$m,n\to\infty$,
so that~$\left(x_{n}\right)$ is a Cauchy sequence and thus converges
to a limit~$x$. By continuity of~$L$, $Lx=y$ and~$L$ is indeed
one to one.\end{proof}
\begin{defn}
An application~$c$ from~$\mathcal{Z}$ to~$\mathcal{Z}$ is a
\emph{conjugation} if and only if it satisfies the following conditions.
\begin{enumerate}
\item $c$ si $\mathbb{R}$-linear.
\item $c^{2}=I_{\mathcal{Z}}$.
\item For all~$z_{1}$, $z_{2}$ in~$\mathcal{Z}$, $\left\langle cz_{1},z_{2}\right\rangle =\left\langle cz_{2},z_{1}\right\rangle $.
\end{enumerate}
\end{defn}
\begin{rem}
It follows from the third condition in this definition that a conjugation
is antilinear.
\end{rem}

One may define different conjugations on the same Hilbert space over~$\mathbb{C}$
(even for a one dimensional Hilbert space). As an example one can
consider a Hilbert basis~$\left(e_{j}\right)$ and define the application~$c:\sum_{j}\alpha_{j}e_{j}\mapsto\sum_{j}\overline{\alpha_{j}}e_{j}$.
\begin{defn}
\label{def:real-space-associated}Let~$c$ be a conjugation on the
Hilbert space~$\mathcal{Z}$. The \emph{real} and \emph{imaginary}
parts of a vector~$z\in\mathcal{Z}$ (with respect to the conjugation
$c$) are defined as\[
\Re z:=\frac{z+cz}{2}\quad\mbox{and}\quad\Im z:=\frac{z-cz}{2i}\,.\]
They verify~$z=\Re z+i\Im z$. The space~$E_{\mathbb{R}}^{c}:=\Re\mathcal{Z}=\Im\mathcal{Z}$
is a subspace of~$\mathcal{Z}$ as $\mathbb{R}$-vector space, $\left\langle \cdot,\cdot\right\rangle $
restricted to~$E_{\mathbb{R}}^{c}$ is a real scalar product and~$E=E_{\mathbb{R}}^{c}\oplus iE_{\mathbb{R}}^{c}$.

Let~$f$ be a $\mathbb{R}$-linear application on~$\mathcal{Z}$,
then we can define the applications from~$E_{\mathbb{R}}^{c}$ to
itself\[
\begin{array}{cc}
\alpha:z\mapsto\Re f\left(z\right)\,,\quad & \gamma:z\mapsto\Re f\left(iz\right)\,,\\
\beta:z\mapsto\Im f\left(z\right)\,,\quad & \delta:z\mapsto\Im f\left(iz\right)\,.\end{array}\]
Then, if~$a,\, b\in E_{\mathbb{R}}^{c}$, then~$f\left(a+ib\right)=\alpha\left(a\right)+i\beta\left(a\right)+\gamma\left(ib\right)+i\delta\left(ib\right)$,
and~$f$ can be represented as an application on~$E_{\mathbb{R}}^{c}\times E_{\mathbb{R}}^{c}$
by the matrix\[
\left(\begin{array}{cc}
\alpha & \gamma\\
\beta & \delta\end{array}\right)\,.\]
The following relations hold with the above sign if~$f$ is $\mathbb{C}$-linear
and with the below sign if~$f$ is $\mathbb{C}$-antilinear: $\beta=\mp\gamma$
and~$\alpha=\pm\delta$ and~$f^{*}$ is represented by the matrix~$\left(\begin{smallmatrix}\alpha^{T} & \mp\beta^{T}\\
\beta^{T} & \pm\alpha^{T}\end{smallmatrix}\right)$.
\end{defn}
We want to show a reduction result for the symplectomorphisms in the
spirit of the polar decomposition, in the case of an implementable
symplectomorphism (see Definition~\ref{def:implementable-symplectomorphism}
and Theorem~\ref{thm:Shale}).
\begin{thm}
\label{thm:decomp_symplecto}Let~$T$ be an implementable symplectomorphism.
Then%
{}\[
T=ue^{c\rho}\]
where
\begin{itemize}
\item $u$ is a unitary operator,
\item $c$ is a conjugation,
\item $\rho$ is a Hilbert-Schmidt, self-adjoint, non-negative operator
commuting with~$c$.
\end{itemize}
\end{thm}
\begin{rem}
The operator~$u$ is the unitary operator of the polar decomposition~$L=u\left|L\right|$
of the $\mathbb{C}$-linear part of~$T$. The conjugation~$c$ is
a specific conjugation associated with~$L$ and will be constructed
during the proof and~$\rho=\arg\cos\left|L\right|$.
\end{rem}

\begin{proof}
Let us write~$T=L+A$ with~$L$ $\mathbb{C}$-linear and~$A$ $\mathbb{C}$-antilinear.
With~$L=u\left|L\right|$ the polar decomposition of~$L$ we get~$T=u\left(\left|L\right|+u^{*}A\right)$
so that it is enough to show the two next lemmas.\end{proof}
\begin{lem}
\label{lem:antiliear-application-reduction}Let~$\left(E,\left\langle \cdot,\cdot\right\rangle \right)$
be a finite-dimensional Hilbert space over~$\mathbb{C}$. Let~$f:E\to E$
be a $\mathbb{C}$-antilinear application such that\[
ff^{*}=I_{E}\qquad\mbox{and}\qquad f=f^{*}\,.\]
Then there exists an orthonormal basis~$\left(u_{j}\right)$ of~$E$
such that\[
\forall j,\, f\left(u_{j}\right)=u_{j}\,.\]
\end{lem}
\begin{proof}
Let us consider an arbitrary conjugation~$c_{0}$ on $E$ and the~$\left(\begin{smallmatrix}\alpha & \beta\\
\beta & -\alpha\end{smallmatrix}\right)$ {}``matrix'' of~$f$ (as a $\mathbb{R}$-linear operator) on $E=E_{\mathbb{R}}^{c_{0}}\oplus iE_{\mathbb{R}}^{c_{0}}$
identified with~$E_{\mathbb{R}}^{c_{0}}\times E_{\mathbb{R}}^{c_{0}}$.
The matrix associated to~$f^{*}$ is~$ $$\left(\begin{smallmatrix}\alpha^{T} & \beta^{T}\\
\beta^{T} & -\alpha^{T}\end{smallmatrix}\right)$ so that the relation~$f=f^{*}$ gives~$\alpha=\alpha^{T}$ and~$\beta=\beta^{T}$.
From~$ff^{*}=I_{E}$ we deduce~$\alpha^{2}+\beta^{2}=\mbox{Id}$
and~$\alpha\beta=\beta\alpha$. We can thus diagonalize simultaneously~$\alpha$
and~$\beta$, and so in a convenient basis of $E_{\mathbb{R}}^{c_{0}}$
the matrix of $f$ is of the form\[
\left(\begin{array}{ccc|ccc}
\ddots &  & 0 & \ddots &  & 0\\
 & \lambda_{j}^{\alpha} &  &  & \lambda_{j}^{\beta}\\
0 &  & \ddots & 0 &  & \ddots\\
\hline \ddots &  & 0 & \ddots &  & 0\\
 & \lambda_{j}^{\beta} &  &  & -\lambda_{j}^{\alpha}\\
0 &  & \ddots & 0 &  & \ddots\end{array}\right)\,.\]
We can thus confine ourself to the case of a space~$E$ of complex
dimension~$1$ and of~$f$ with a matrix of the form~$\left(\begin{smallmatrix}\alpha & \beta\\
\beta & -\alpha\end{smallmatrix}\right)$ with~$\alpha$ and~$\beta$ real numbers. We search a normalized
vector~$z=\left(\begin{smallmatrix}\cos\theta\\
\sin\theta\end{smallmatrix}\right)=\left(\begin{smallmatrix}x\\
y\end{smallmatrix}\right)$ and a real~$\lambda$ such that~$f\left(z\right)=\lambda z$, i.e.
\begin{eqnarray*}
\lambda\left(\begin{smallmatrix}x\\
y\end{smallmatrix}\right) & = & \left(\begin{smallmatrix}\phantom{-\mbox{}}\alpha x+\beta y\\
-\alpha y+\beta x\end{smallmatrix}\right)=\left(\begin{smallmatrix}x & y\\
-y & x\end{smallmatrix}\right)\left(\begin{smallmatrix}\alpha\\
\beta\end{smallmatrix}\right)\\
 & = & \sqrt{\alpha^{2}+\beta^{2}}\left(\begin{smallmatrix}\cos\theta & \sin\theta\\
-\sin\theta & \cos\theta\end{smallmatrix}\right)\left(\begin{smallmatrix}\cos\phi\\
\sin\phi\end{smallmatrix}\right)=\sqrt{\alpha^{2}+\beta^{2}}\left(\begin{smallmatrix}\cos\left(\phi-\theta\right)\\
\sin\left(\phi-\theta\right)\end{smallmatrix}\right)\end{eqnarray*}
 so that if we choose~$\theta$ such that~$\phi-\theta=\theta$
we get the desired result with~$\lambda=\sqrt{\alpha^{2}+\beta^{2}}$.
Finally, from~$ff^{*}=I_{E}$ we deduce that~$\lambda=1$ and the
result follows.\end{proof}
\begin{lem}
Let~$T=L+A$ be an implementable symplectomorphism with~$L$ $\mathbb{C}$-linear
self-adjoint and positive, $A$ $\mathbb{C}$-antilinear. 

Then~$L$ and~$A$ commute, there exist a conjugation~$c$ commuting
with~$L$ and~$A$ such that~$Ac$ is self-adjoint and non-negative
and\[
T=e^{c\rho}\]
with~$\rho=\arg\cosh L=\arg\sinh\left(Ac\right)$ a Hilbert-Schmidt,
non-negative and self-adjoint operator commuting with~$c$.\end{lem}
\begin{proof}
As~$AA^{*}\in\mathcal{L}_{1}\left(\mathcal{Z}\right)$, $AA^{*}=\sum_{j}\lambda_{j}^{2}\left|e_{j}\right\rangle \left\langle e_{j}\right|$,
with~$\lambda_{j}\in\mathbb{R}$ and~$\sum_{j}\lambda_{j}^{2}<\infty$,
from~$L^{2}=I_{\mathcal{Z}}+AA^{*}$ we deduce~$L^{2}=\sum_{j}\mu_{j}^{2}\left|e_{j}\right\rangle \left\langle e_{j}\right|$
with~$\mu_{j}=\sqrt{1+\lambda_{j}^{2}}$ and thus~$L=\sum_{j}\mu_{j}\left|e_{j}\right\rangle \left\langle e_{j}\right|$.

From the equivalent characterizations of a symplectomorphism we get\[
L^{2}-AA^{*}=I_{\mathcal{Z}}\quad\mbox{and}\quad L^{2}-A^{*}A=I_{\mathcal{Z}}\]
multiplying the first equality on the right and the second on the
left by~$A$ and computing the difference we get~$\left[L^{2},A\right]=0$.
As~$L$ is self-adjoint and positive one can use the functional calculus
and~$L=\sqrt{L^{2}}$ to obtain~$\left[L,A\right]=0$. 

From~$\left[L,A\right]=0$, $L=L^{*}$ and the characterizations
of a symplectomorphism, we also get~$AL=LA=L^{*}A=A^{*}L$ so that~$\left(A-A^{*}\right)L=0$
and from the invertibility of~$L$ one deduces~$A=A^{*}$.

The proper subspaces associated with~$L$ and~$\ker\left(L-\mu I_{\mathcal{Z}}\right)$,
are thus stable by the action of~$A$ (and finite-dimensional). We
also remark that on~$\ker\left(L-\mu I_{\mathcal{Z}}\right)$, $AA^{*}=L^{2}-I_{\mathcal{Z}}=(\mu^{2}-1)I_{\mathcal{Z}}$,
so that two cases are possible:\[
\mu=1\,,\,\mbox{then}\, A=0\qquad\mbox{or}\qquad\mu>1\,,\,\mbox{then}\,\frac{1}{\sqrt{\mu^{2}-1}}A\frac{1}{\sqrt{\mu^{2}-1}}A^{*}=I_{\mathcal{Z}}\,.\]

We apply Lemma~\ref{lem:antiliear-application-reduction} to the
$\mathbb{C}$-antilinear applications induced by the applications
$A/\sqrt{\mu^{2}-1}$ on the Hilbert spaces~$\ker\left(L-\mu I_{\mathcal{Z}}\right)$.
This provides us with a Hilbert basis~$\left(e_{j}\right)$ of~$\mathcal{Z}$
which diagonalizes both~$L$ and~$A$. We can also define a conjugation~$c\left(\sum_{j}\alpha_{j}e_{j}\right)=\sum_{j}\overline{\alpha_{j}}e_{j}$.
This conjugation commutes with~$L$ and~$A$, and~$Ac$ is clearly
a non-negative self-adjoint operator and so is necessarily~$\sqrt{AA^{*}}$.
We finally get for every vector~$e_{j}$ of the basis the relations~$Le_{j}=\mu_{j}e_{j}$
and~$Ae_{j}=\lambda_{j}e_{j}$ with~$\mu_{j}^{2}-\lambda_{j}^{2}=1$,
and thus one can define~$\rho_{j}=\arg\cosh\mu_{j}$ ($\rho_{j}=\arg\sinh\lambda_{j}$
as $\lambda_{j}\geq0$) and so we can define~$\rho=\arg\cosh L=\arg\sinh Ac$
so that~$T=e^{c\rho}$.
\end{proof}

\section{Relations between Weyl and Wick symbols in finite dimension}

%
{}We want to use the relation between the Weyl and Wick symbols associated
to a same Wick polynomial in finite dimension, working with~$\mathcal{Z}=\mathbb{C}^{r}$
we have \[
b=\frac{1}{\left(\pi\varepsilon/2\right)^{r}}\breve{b}*e^{-\frac{\left|z\right|^{2}}{\varepsilon/2}}\]
where~$b$ is the Wick symbol and~$\breve{b}$ is the Weyl symbol
and~$b^{Wick}=\breve{b}^{Weyl}$. We want to get rid of the convolution
and for this we use the Fourier transform\[
\mathcal{F}f\left(x'\right)=\frac{1}{\left(2\pi\right)^{r}}\int_{\mathbb{R}^{2r}}e^{-ix.x'}f\left(x\right)dx\]
where~$x,x'\in\mathbb{R}^{2r}\cong\mathbb{C}^{r}$. The inverse Fourier
transform is then\[
\mathcal{F}^{-1}f\left(x\right)=\frac{1}{\left(2\pi\right)^{r}}\int_{\mathbb{R}^{2r}}e^{ix.x'}f\left(x'\right)dx'\,.\]
We can use the formulae\begin{eqnarray*}
\mathcal{F}\left(f*g\right) & = & \left(2\pi\right)^{r}\mathcal{F}f.\mathcal{F}g\\
\mathcal{F}\left[e^{-\alpha\frac{\left|x\right|^{2}}{2}}\right]\left(x'\right) & = & \frac{1}{\alpha^{r}}e^{-\frac{\left|x'\right|^{2}}{2\alpha}}\\
\mathcal{F}^{-1}\left(x\times\cdot\right)\mathcal{F} & = & D_{x}\,.\end{eqnarray*}
We then obtain with~$m=2n$\begin{eqnarray*}
\mathcal{F}b\left(z'\right) & = & \frac{\left(2\pi\right)^{r}}{\left(\pi\varepsilon/2\right)^{r}}\mathcal{F}\left[e^{-\frac{\left|z\right|^{2}}{\varepsilon/2}}\right]\mathcal{F}\breve{b}\left(z'\right)\\
 & = & \left(\frac{4}{\varepsilon}\right)^{r}\left(\frac{\varepsilon}{4}\right)^{r}e^{-\frac{\varepsilon}{8}\left|z'\right|^{2}}\mathcal{F}\breve{b}\left(z'\right)\\
 & = & e^{-\frac{\varepsilon}{8}\left|z'\right|^{2}}\mathcal{F}\breve{b}\left(z'\right)\end{eqnarray*}
and%
{}\begin{eqnarray*}
b & = & \mathcal{F}^{-1}e^{-\frac{\varepsilon}{8}\left|z'\right|^{2}}\mathcal{F}\breve{b}\\
 & = & e^{-\frac{\varepsilon}{8}\mathcal{F}^{-1}\left|z'\right|^{2}\mathcal{F}}\breve{b}\\
 & = & e^{\frac{\varepsilon}{2}\partial_{z}.\partial_{\bar{z}}}\breve{b}\end{eqnarray*}
using the fact that \[
\mathcal{F}^{-1}\left|z'\right|^{2}\mathcal{F}=D_{\left(x,\xi\right)}^{2}=-4\times\frac{1}{2}\left(\partial_{x}-i\partial_{\xi}\right).\frac{1}{2}\left(\partial_{x}+i\partial_{\xi}\right)=-4\partial_{z}.\partial_{\bar{z}}\,.\]

It is clear that if~$\breve{b}$ is a polynomial in~$\mathcal{P}_{\leq m}\left(\mathcal{Z}\right)$,
then~$b$ is in this class of polynomials, as we can see deriving
the convolution product. We want to show that the application\begin{eqnarray*}
\mathcal{P}_{\leq m}\left(\mathcal{Z}\right) & \to & \mathcal{P}_{\leq m}\left(\mathcal{Z}\right)\\
\breve{b} & \mapsto & b=\frac{1}{\left(\pi\varepsilon/2\right)^{n}}\breve{b}*e^{-\frac{\left|z\right|^{2}}{\varepsilon/2}}\end{eqnarray*}
is a bijection. As the dimension of~$\mathcal{Z}$ is finite, the
dimension of~$\mathcal{P}_{\leq m}\left(\mathcal{Z}\right)$ is finite
and it is enough to show the injectivity of this application. For
this we want to justify that on the part of main degree this application
is the identity. This is obvious from the following facts:
\begin{itemize}
\item $\partial_{\bar{z}}^{q}\partial_{z}^{p}b=\frac{1}{\left(\pi\varepsilon/2\right)^{r}}\partial_{\bar{z}}^{q}\partial_{z}^{p}\breve{b}*e^{-\frac{\left|z\right|^{2}}{\varepsilon/2}}$
\item this application is the identity on the constants.
\end{itemize}
Thus we can also consider the reverse application that we will improperly
note \[
\breve{b}=e^{-\frac{\varepsilon}{2}\partial_{z}.\partial_{\bar{z}}}b\,.\]

\section{Symplectic Fourier transform\label{sec:Symplectic-Fourier-transform}}

Let us then consider the symplectic Fourier transform on~$L^{2}\left(\mathbb{C}^{d};\mathbb{C}\right)\equiv L^{2}\left(\mathbb{R}^{2d}\right)$
with~$z=x+iy$, defined by\[
\mathcal{F}^{\sigma}\left(f\right)\left(z\right)=\int e^{i2\pi\sigma\left(z,z'\right)}f\left(z'\right)L\left(dz'\right)\]
with~$\sigma\left(z,z'\right)=\Im\left\langle z,z'\right\rangle =\Im\left[\left\langle x,x'\right\rangle +\left\langle y,y'\right\rangle +i\left\langle x,y'\right\rangle -i\left\langle y,x'\right\rangle \right]$
and~$L$ denotes the Lebesgue measure. We list here some properties
of the symplectic Fourier transform.

\begin{enumerate}
\item Inverse.\[
\left(\mathcal{F}^{\sigma}\right)^{-1}=\mathcal{F}^{\sigma}\]

\item Convolution. \[
\mathcal{F}^{\sigma}\left(f*g\right)=\mathcal{F}^{\sigma}f.\mathcal{F}^{\sigma}g\]

\item Composition with a symplectic transformation.\label{sub:Composition-F-symplecto}
Let~$T$ be a symplectomorphism, then\[
\mathcal{F}^{\sigma}\left[f\left(T\cdot\right)\right]\left(z\right)=\mathcal{F}^{\sigma}\left[f\right]\left(Tz\right)\,.\]

\item Gaussians. For~$a>0$,  \[
\mathcal{F}^{\sigma}\left[e^{-a\left|\cdot\right|^{2}}\right]\left(z\right)=\left(\frac{\pi}{a}\right)^{d}e^{-\pi^{2}\left|z\right|^{2}/a}\,.\]

\item Derivation. We consider the derivations~$\partial_{z}=\frac{1}{2}\left(\partial_{x}-i\partial_{y}\right)$
and~$\partial_{\bar{z}}=\frac{1}{2}\left(\partial_{x}+i\partial_{y}\right)$
then\[
-\frac{1}{\pi}\partial_{z}.z_{0}=\mathcal{F}^{\sigma}\left(\bar{z}.z_{0}\times\right)\mathcal{F}^{\sigma}\quad\mbox{and}\quad\frac{1}{\pi}\bar{z}_{0}.\partial_{\bar{z}}=\mathcal{F}^{\sigma}\left(\bar{z}_{0}.z\times\right)\mathcal{F}^{\sigma}\,.\]
 \end{enumerate}
\begin{acknowledgement*}
The author would like to thank Francis Nier and Zied Ammari for profitable
discussions.
\end{acknowledgement*}

\bibliographystyle{plain}
\bibliography{biblio_lemme_hepp}

\begin{thebibliography}{10}

\bibitem{MR2465733}
Zied Ammari and Francis Nier.
\newblock Mean field limit for bosons and infinite dimensional phase-space
  analysis.
\newblock {\em Ann. Henri Poincar\'e}, 9(8):1503--1574, 2008.

\bibitem{MR1178936}
John~C. Baez, Irving~E. Segal, and Zheng-Fang Zhou.
\newblock {\em Introduction to algebraic and constructive quantum field
  theory}.
\newblock Princeton Series in Physics. Princeton University Press, Princeton,
  NJ, 1992.

\bibitem{MR0208930}
Feliks~A. Berezin.
\newblock {\em The method of second quantization}.
\newblock Academic Press, New York, 1966.

\bibitem{MR1186643}
Feliks~A. Berezin and Mikhail~A. Shubin.
\newblock {\em The {S}chr\"odinger equation}, volume~66 of {\em Mathematics and
  its Applications (Soviet Series)}.
\newblock Kluwer Academic Publishers Group, Dordrecht, 1991.

\bibitem{MR2297950}
Laurent Bruneau and Jan Derezi{\'n}ski.
\newblock Bogoliubov {H}amiltonians and one-parameter groups of {B}ogoliubov
  transformations.
\newblock {\em J. Math. Phys.}, 48(2):022101, 24, 2007.

\bibitem{MR1690026}
Monique Combescure, James Ralston, and Didier Robert.
\newblock A proof of the {G}utzwiller semiclassical trace formula using
  coherent states decomposition.
\newblock {\em Comm. Math. Phys.}, 202(2):463--480, 1999.

\bibitem{MR2221699}
Monique Combescure and Didier Robert.
\newblock Quadratic quantum {H}amiltonians revisited.
\newblock {\em Cubo}, 8(1):61--86, 2006.

\bibitem{MR883643}
Hans~L. Cycon, Richard~G. Froese, Werner Kirsch, and Barry Simon.
\newblock {\em Schr\"odinger operators with application to quantum mechanics
  and global geometry}.
\newblock Texts and Monographs in Physics. Springer-Verlag, Berlin, study
  edition, 1987.

\bibitem{MR0391794}
William~G. Faris and Richard~B. Lavine.
\newblock Commutators and self-adjointness of {H}amiltonian operators.
\newblock {\em Comm. Math. Phys.}, 35:39--48, 1974.

\bibitem{MR2291792}
J{\"u}rg Fr{\"o}hlich, Sandro Graffi, and Simon Schwarz.
\newblock Mean-field- and classical limit of many-body {S}chr\"odinger dynamics
  for bosons.
\newblock {\em Comm. Math. Phys.}, 271(3):681--697, 2007.

\bibitem{MR530915}
Jean Ginibre and Giorgio Velo.
\newblock The classical field limit of scattering theory for nonrelativistic
  many-boson systems. {I}.
\newblock {\em Comm. Math. Phys.}, 66(1):37--76, 1979.

\bibitem{MR539736}
Jean Ginibre and Giorgio Velo.
\newblock The classical field limit of scattering theory for nonrelativistic
  many-boson systems. {II}.
\newblock {\em Comm. Math. Phys.}, 68(1):45--68, 1979.

\bibitem{MR602197}
Jean Ginibre and Giorgio Velo.
\newblock The classical field limit of nonrelativistic bosons. {I}. {B}orel
  summability for bounded potentials.
\newblock {\em Ann. Physics}, 128(2):243--285, 1980.

\bibitem{MR605198}
Jean Ginibre and Giorgio Velo.
\newblock The classical field limit of nonrelativistic bosons. {II}.
  {A}symptotic expansions for general potentials.
\newblock {\em Ann. Inst. H. Poincar\'e Sect. A (N.S.)}, 33(4):363--394, 1980.

\bibitem{MR2575484}
Manoussos~G. Grillakis, Matei Machedon, and Dionisios Margetis.
\newblock Second-order corrections to mean field evolution of weakly
  interacting bosons. {I}.
\newblock {\em Comm. Math. Phys.}, 294(1):273--301, 2010.

\bibitem{MR0332046}
Klaus Hepp.
\newblock The classical limit for quantum mechanical correlation functions.
\newblock {\em Comm. Math. Phys.}, 35:265--277, 1974.

\bibitem{MR0249698}
Einar Hille.
\newblock {\em Lectures on ordinary differential equations}.
\newblock Addison-Wesley Publ. Co., Reading, Mass.-London-Don Mills, Ont.,
  1969.

\bibitem{MR1339714}
Lars H{\"o}rmander.
\newblock Symplectic classification of quadratic forms, and general {M}ehler
  formulas.
\newblock {\em Math. Z.}, 219(3):413--449, 1995.

\bibitem{MR0279626}
Tosio Kato.
\newblock Linear evolution equations of ``hyperbolic'' type.
\newblock {\em J. Fac. Sci. Univ. Tokyo Sect. I}, 17:241--258, 1970.

\bibitem{MR0161185}
Jan Kisy{{\'n}}ski.
\newblock Sur les op{\'e}rateurs de {G}reen des probl{\`e}mes de {C}auchy
  abstraits.
\newblock {\em Studia Math.}, 23:285--328, 1963/1964.

\bibitem{MR628382}
Lutz Polley, G.~Reents, and Raymond~F. Streater.
\newblock Some covariant representations of massless boson fields.
\newblock {\em J. Phys. A}, 14(9):2479--2488, 1981.

\bibitem{MR0493420}
Michael Reed and Barry Simon.
\newblock {\em Methods of modern mathematical physics. {II}. {F}ourier
  analysis, self-adjointness}.
\newblock Academic Press [Harcourt Brace Jovanovich Publishers], New York,
  1975.

\bibitem{MR2530155}
Igor Rodnianski and Benjamin Schlein.
\newblock Quantum fluctuations and rate of convergence towards mean field
  dynamics.
\newblock {\em Comm. Math. Phys.}, 291(1):31--61, 2009.

\bibitem{MR0137504}
David Shale.
\newblock Linear symmetries of free boson fields.
\newblock {\em Trans. Amer. Math. Soc.}, 103:149--167, 1962.

\end{thebibliography}

\end{document}